\documentclass[a4paper]{article}
\usepackage[all]{xy}\usepackage[latin1]{inputenc}        %accents
\usepackage[dvips]{graphics,graphicx}
\usepackage{amsfonts,amssymb,amsmath,color,mathrsfs, amstext}
\usepackage{amsbsy, amsopn, amscd, amsxtra, amsthm,authblk}
\usepackage{enumerate,algorithmicx,algorithm}
\usepackage{algpseudocode}
\usepackage{upref}
\usepackage{geometry}
\usepackage{amsthm,amsmath,amssymb}
\usepackage{relsize}
\usepackage{mathrsfs}
\usepackage{bm}
\usepackage{exscale}
\usepackage{graphicx}
\usepackage{tabularx}
\usepackage{threeparttable,multirow,dcolumn,booktabs}
\usepackage{algorithmicx,algorithm}
\usepackage{tabularx}
\usepackage{caption}
\usepackage{float}
\usepackage{bm}% bold math
\usepackage{caption}
\usepackage{yhmath}
\usepackage{booktabs}
\usepackage{subcaption}
\usepackage{multirow}
\usepackage{makecell}
\usepackage[colorlinks,
            linkcolor=red,
            anchorcolor=red,
            citecolor=red
            ]{hyperref}

\numberwithin{equation}{section}

\numberwithin{equation}{section}

\usepackage{soul}
\usepackage{enumitem}
\usepackage{bbm}
\newtheorem{theorem}{Theorem}[section]%[section]
\newtheorem{lemma}[theorem]{Lemma}%[section]
\usepackage{comment}
\usepackage{exscale}
\usepackage{relsize}
\usepackage{booktabs}
\usepackage{longtable}
\usepackage{multirow}
\usepackage{makecell}

\numberwithin{equation}{section}

\begin{document}

\title{Sum-of-Gaussians tensor neural networks for high-dimensional Schr\"{o}dinger equation}

%\date{\today}

\author[{\$}1]{Qi Zhou}

\author[{\$}2]{Teng Wu}

\author[2]{Jianghao Liu}

\author[2]{Qingyuan Sun}

\author[3]{Hehu Xie\thanks{hhxie@lsec.cc.ac.cn}}

\author[1,4,5]{Zhenli Xu\thanks{xuzl@sjtu.edu.cn}}

\affil[1]{School of Mathematical Sciences, Shanghai Jiao Tong University, 200240, Shanghai, China}

\affil[2]{Zhiyuan College, Shanghai  Jiao  Tong  University,  Shanghai 200240, China}

\affil[3]{SKLMS, NCMIS, Institute of Computational Mathematics, Academy of Mathematics and Systems
Science, Chinese Academy of Sciences, Beijing 100190,
China,  and School of Mathematical Sciences, University of Chinese Academy of Sciences, Beijing 100049, China}

\affil[4]{MOE-LSC and CMA-Shanghai, Shanghai Jiao Tong University, 200240, Shanghai, China}

\affil[5]{SOG AI-Technology Co. Ltd., Shanghai.}

\affil[{$\$$}]{These authors contributed equally to this work}

\date{}
\maketitle

%    Abstract is required.
\begin{abstract}
We propose an accurate, efficient, and low-memory sum-of-Gaussians tensor neural network (SOG-TNN) algorithm for solving the high-dimensional Schr\"{o}dinger equation. The SOG-TNN  utilizes a low-rank tensor product representation of the solution to overcome the curse of dimensionality associated with high-dimensional integration. To handle the Coulomb interaction, we introduce an SOG decomposition to approximate the interaction kernel such that it is dimensionally separable, leading to a tensor representation with rapid convergence. We further develop a range-splitting scheme that partitions the Gaussian terms into short-, long-, and mid-range components. They are treated with the asymptotic expansion, the low-rank Chebyshev expansion, and the model reduction with singular-value decomposition, respectively, significantly reducing the number of two-dimensional integrals in computing electron-electron interactions. The SOG decomposition well resolves the computational challenge due to the singularity of the Coulomb interaction, leading to an efficient algorithm for the high-dimensional problem under the TNN framework. 
Numerical results demonstrate the outstanding performance of the new method, revealing that the SOG-TNN is a promising way for accurately tackling quantum systems.


\end{abstract}

%% main text
\section{Introduction}\label{sec::intro}
The solution of high-dimensional many-body systems with quantum accuracy represents a long-standing challenge in scientific computing. One central theme lies in the efficient solution of the Schr\"{o}dinger equation of high dimensions \cite{Schrodinger1926, Dirac1929}, which attracts broad interest in areas such as materials science \cite{POKLUDA2015127,Taureau2024}, quantum chemistry \cite{Hermann2020,Hermann2023,pfau2024excited}, and quantum optics \cite{Bekenstein2020,Liu2023,Douglas2015,Kira2011}.
The greatest difficulty in solving the high-dimensional Schr\"{o}dinger equation lies in the curse of dimensionality in storing and computing the wavefunction, where the computational complexity often grows exponentially with the number of electrons. Challenges also include the strong entanglement between electrons due to the Pauli exclusion principle, and the long-range nature and the origin singularity of the Coulomb kernel, posing additional difficulties in numerical calculations.

In recent decades, a variety of approaches have been developed to tackle above challenges. One class of methods is the wavefunction approach, which aims at a direct solution of the Schr\"{o}dinger equation. These include configuration interaction (CI)~\cite{pople1987quadratic,rask2021toward}, coupled cluster (CC)~\cite{purvis1982full,nagy2019approaching}, and quantum Monte Carlo (QMC)~\cite{zhang2003quantum,foulkes2001quantum,godoy2025software}. However, their application to large-scale systems is often constrained by high computational cost. Deterministic post-Hartree-Fock methods like CI and CC typically encounter steep polynomial or exponential scaling with respect to the number of electrons. Conversely, QMC scales more favorably, but it requires extensive statistical sampling to suppress noise. Approximate schemes based on the density functional theory \cite{Kohn1965,Martin2020,becke2014perspective} offer a viable path to reduce computational complexity. However, this class of methods presents an intrinsic trade-off between computational accuracy and efficiency, a challenge often referred to as the `Jacob's Ladder' \cite{Car2016Fixing}. Recently, the integration of deep neural network (DNN) architectures with quantum physics has sparked a revolution in representing high-dimensional functions~\cite{Han2017, E2018, Raissi2019}. Aided by stochastic sampling schemes such as variational Monte Carlo (VMC)~\cite{ceperley1977monte,bressanini1999between}, DNN-based solvers have achieved extraordinary success, effectively overcoming the curse of dimensionality and establishing new benchmarks for large-scale systems with unprecedented accuracy~\cite{Hermann2020,Hermann2023,pfau2024excited,choo2020fermionic,Carleo2017, Han2020, Pfau2020}. These methods have fundamentally expanded the boundaries of \textit{ab initio} calculations. However, by design, VMC-based approaches rely on stochastic optimization and estimation. While powerful, this inherently precludes the benefits of deterministic algorithms, such as noise-free evaluation, strict error bounds, and rapid spectral convergence, which are desirable for high-precision theoretical studies.

To achieve high precision without stochastic noise, one must revisit the potential of deterministic scaling. As a rigorous extension of classical Hartree-product-based approaches, the sparse grid methods~\cite{griebel2006wavelet,griebel2007sparse,griebel2009tensor} exploit sparsity patterns derived from the mixed-derivative regularity analysis of the wavefunction, thereby drastically reducing the degrees of freedom while preserving convergence rates essentially independent of the number of electrons~\cite{yserentant2004regularity,Yserentant2005Sparse,Yserentant2007Hyperbolic,Yserentant2012Mixed}. In practice, however, the efficiency of this approach depends fundamentally on the fixed underlying basis. While sparse grids theoretically achieve dimension-independent convergence rates, the curse of dimensionality often resurfaces through the uncontrolled growth of error prefactors. The recently developed Tensor Neural Network (TNN) architecture~\cite{Wang2024Tensor,wang2024multi} can be regarded as the re-extension of sparse grid method, providing a robust solution. The TNN effectively reinvents the sparse grid as an adaptive and learnable ansatz capturing a significantly sparser representation. By replacing the rigid, fixed basis with powerful, non-linear deep neural networks, TNN synergizes the high-dimensional expressivity with the deterministic structure of tensor decompositions and has been successfully applied to solving a wide range of high-dimensional problems~\cite{Yu2026Tensor,Hu2024Tackling,Lin2025Tensor,Lin2025Solving,Wang2025Tensor}. This hybridization transforms the many-body problem into a deterministic optimization task, enabling, for the first time, high-precision, noise-free solutions that enjoy the scalability of deep learning without the stochastic limitations.
Nevertheless, due to the presence of the Coulomb interaction in the Schr\"{o}dinger equation, the challenge of high-dimensional integration persists when the wavefunction is represented by the TNN. The TNN for solving the Schr\"odinger equation was first explored in \cite{202209v2}, where the Coulomb kernel is expanded by spherical harmonic expansion with associated Legendre polynomials. However, the convergence of the tensor representation with the spherical harmonics is slow due to the kernel singularity. To achieve high efficiency, an optimized representation for the Coulomb interaction is crucial for the TNN solution for the high-dimensional Schr\"{o}dinger equation.

In this paper, we propose a novel method based on the TNN and a sum-of-Gaussians (SOG) decomposition, SOG-TNN, for the efficient computation of the eigenvalue problem for the high-dimensional Schr\"{o}dinger equation. The SOG approximation on the Coulomb interaction offers the advantages of a uniform error distribution and rapid convergence \cite{beylkin2005,beylkin2010,LIANG2025101759}. Crucially, each Gaussian term is separable across spatial dimensions, a property that aligns perfectly with the tensor product structure of the TNN framework. This separability enables the decomposition of the six-dimensional interaction integral into a product of three two-dimensional integrals. Subsequently, we employ this structure by designing distinct low-rank methods tailored to the different bandwidths of the terms in the SOG expansion. This allows for the efficient computation of contributions from these varying interaction ranges. Numerical results demonstrate that the SOG-TNN method can compute the ground-state energies for the helium  and lithium systems with over $10^{-7}$ accuracy, and more efficient than the TNN with the tensor representation via spherical harmonics (SHE-TNN). For the beryllium atom, the SOG-TNN method achieves $10^{-5}$ accuracy while using only one-tenth of the memory of a single GPU.  This represents an accuracy improvement of three orders of magnitude over the SHE-TNN method that utilized the full memory of the GPU. Besides, the comparison results indicate that SOG-TNN reduces the basis size by more than two orders of magnitude while maintaining the accuracy of the baseline and preserving accurate antisymmetry.
Notably, although our method utilizes a machine learning framework, the learning error of the SOG-TNN approach is negligible compared to the approximation error in application scenarios \cite{wang2024solving,wang2024multi}. This allows for the explicit determination of the parameter settings required to achieve a desired precision. The variational principle underlying the Galerkin framework ensures that the method is naturally applicable to high-precision calculations of various excited states \cite{WANG2024112928}, highlighting its strong potential for a wide range of future applications.

The remainder of this paper is organized as follows. In Section~\ref{sec::VP}, we review the high-dimensional Schr\"{o}dinger equation. Section~\ref{sec::SOG-TNN} presents the details of the proposed SOG-TNN algorithm. Section~\ref{sec::Range_split} develops a range-splitting strategy to accelerate the calculation. Section~\ref{sec::error} is devoted to the parameter selection strategy for the SOG-TNN method by estimating the convergence of errors. 
Numerical results for several atomic systems are presented in Section~\ref{sec::num_example}. 
Concluding remarks are offered in Section~\ref{sec::conclusion}. 

\section{The Schr\"{o}dinger equation}
\label{sec::VP}
Consider a quantum system with $M$ ions and $N$ electrons. The eigenvalue problem for the Schr\"{o}dinger equation is given by
\begin{equation}
\hat{H}\Psi(\bm{r})=E\Psi(\bm{r}),
\label{TISE}
\end{equation}
where $E$ is the eigenvalue, and $\hat{H}$ is the Hamiltonian operator described by
\begin{equation}
\hat{H} = -\frac{1}{2}\sum_{i=1}^N \Delta_i+ \sum_{i=1}^N \sum_{j=i+1}^N \frac{1}{|\bm{r}_i-\bm{r}_j|} -\sum_{k=1}^M\sum_{i=1}^N \frac{Q_k}{|\bm{r}_i-\bm{R}_k|}+\sum_{k=1}^M\sum_{l=k+1}^M\frac{Q_{k}Q_{l}}{|\bm{R}_{k}-\bm{R}_{l}|}.
\label{Hamiltonian_Operator}
\end{equation}

Here, $\Delta_i$ denotes the Laplace operator of the $i$th electron, $\{Q_k\}_{k=1}^{M}$ are the nuclear charges of ions, and $\bm{r}_i = (x_i,y_i,z_i)$ and 
$\bm{R}_k=(X_k,Y_k,Z_k)$ denote the locations of electrons and ions, respectively. Under the Born-Oppenheimer approximation \cite{Born1927BO}, the wavefunction $\Psi$ is treated as a function of the electronic coordinates only, while the much heavier ions are considered fixed in position. The exponential decay property \cite{Agmon1982} of the wavefunction $\Psi(\bm{r})=\Psi(\bm{r}_1,\cdots,\bm{r}_N)$ allows us to truncate it into a finite domain $\Omega=\Omega_0^{\otimes 3N}$ with $\Omega_0=[-r_c,r_c]$, and consider only the electronic probability distribution therein. The Hamiltonian operator $\hat{H}$ for the high-dimensional Schr\"{o}dinger equation possesses a series of eigenvalues $E_0\le E_1\le \cdots\le E_k\le\cdots$. The primary focus of this work is to solve the ground state energy $E_0$ and its corresponding eigenfunction. The Rayleigh principle \cite{babuska1989finite, babuska1991eigenvalue} states that the ground state energy corresponds to the minimum of the energy functional $\mathscr{E}[\Psi]$,
\begin{equation}\label{eq::Rayleigh}
E_0=\min_{\Psi\in V_{\mathcal{A}}}\mathscr{E}[\Psi],\quad \mathscr{E}[\Psi]=\frac{\langle\Psi|\hat{H}|\Psi\rangle}{\langle\Psi|\Psi\rangle}, 
\end{equation}
where $\langle\cdot|$ and $|\cdot\rangle$ are the Dirac bra and ket notation \cite{dirac1939new}. 
Let $V_\mathcal{A}$ denote the function space of $\mathbb{R}^{3N}$ variables
that satisfies the Pauli exclusion principle, that is, 
for any two identical electrons $\bm{r}_i$ and $\bm{r}_j$, $1\le i< j\le N$, 
with the same spin state, the function $\Psi\in V_\mathcal{A}$ should satisfy 
the antisymmetric condition through the exchange operator $T_{ij}$ that
\begin{equation}\label{eq::Pauli_exclusion}
\frac{\langle T_{ij}\Psi | \Psi \rangle}{\langle \Psi | \Psi \rangle} = -1,\quad T_{ij}\Psi(\bm{r}_1,\cdots,\bm{r}_i,\cdots,\bm{r}_j,\cdots,
\bm{r}_N)=\Psi(\bm{r}_1,\cdots,\bm{r}_j,\cdots,\bm{r}_i,\cdots,\bm{r}_N).
\end{equation}

It is remarked that the general excited-state case, corresponding 
to multiple eigenvalues $E_k$ for $k\ge 1$, the Rayleigh principle 
\cite{babuska1989finite, babuska1991eigenvalue} also provides a similar observation of
\begin{equation}\label{eq::Rayleigh_excited}
E_k = \min_{V_{\mathcal{A}}^{k}} \max_{\Psi \in V_{\mathcal{A}}^{k}} \mathscr{E}[\Psi],
\end{equation}
where $V_{\mathcal{A}}^{k}$ is any $k+1$-dimensional subspace of $V_\mathcal{A}$. 
Therefore, all the techniques for treating the ground state can be easily 
extended to the problem of solving general excited states.

\section{Sum-of-Gaussians tensor neural network}\label{sec::SOG-TNN}
In this section, we propose the TNN method employing the SOG for the tensor 
decomposition of Coulomb interactions, resulting in the efficient SOG-TNN solver 
for the high-dimensional many-body Schr\"{o}dinger equation.

\subsection{TNN architecture}\label{sec::TNN}
Under the Rayleigh principle, the primary computational bottleneck is the high-dimensional numerical integration in Eq.~\eqref{eq::Rayleigh}. The TNN framework provides an effective solution to this challenge \cite{Wang2024Tensor}. Consider a $d$ dimensional function $u(\bm{x})$ for $\bm{x}=(x_1,\cdots,x_d)\in \mathbb{R}^d$ ($d=3N$ in our problem). Let $\bm{\beta}=(\beta_1,\cdots,\beta_d)\in \mathbb{N}^d$ and $\beta=\beta_1+\cdots+\beta_d$. We denote the partial derivative of $u$,
\begin{equation}\label{eq::Multi_index}
\partial^{\bm{\beta}}u=\frac{\partial^{\bm{\beta}}}{\partial\bm{x}^{\bm{\beta}}}u=\frac{\partial^{\beta}}{\partial x_1^{\beta_1}\cdots\partial x_d^{\beta_d}}u.
\end{equation} 
For a bounded tensor-product domain \(\Omega = \bigotimes_{i=1}^{d} \Omega_i\), the Sobolev space \(H^m(\Omega)\) admits an isomorphic tensor decomposition $H^m(\Omega_1 \times \cdots \times \Omega_d) \cong H^m(\Omega_1) \otimes \cdots \otimes H^m(\Omega_d),$ where 
\begin{equation}
H^m(\Omega)=\bigl\{\,u\in L^2(\Omega)\,\bigm|\,
\partial^{\boldsymbol\beta}u\in L^2(\Omega)
\ \text{for all}\ \beta \le m\bigr\},
\end{equation}
and the corresponding norm reads
\begin{equation}
\|u\|_{H^m(\Omega)}:=\left[\int_{\Omega}\sum_{\beta\le m}\left|\partial^{\bm{\beta}}u(\bm{x})\right|^2\mathrm{d}\bm{x}\right]^{1/2}.
\end{equation}
The TNN is a subspace approximation method for high-dimensional functions \cite{Beylkin2002Numerical,Kolda2009Tensor,Hong2020Generalized} by using the isomorphic relation.
Figure~\ref{fig:schema} provides a schematic illustration of the TNN architecture. It comprises $d$ independent feedforward sub-neural networks
\begin{equation}
\bm{\phi}_i(x_i; \bm{\theta}_i) = \left(\phi_{i,1}, \dots, \phi _{i,p}\right)^{\top}, \quad i=1,\dots,d
\end{equation}
with the expression
\begin{equation}\label{eq::FNN_basis}
\bm{\phi}_{i}(x_i;\bm{\theta}_i)=\bigl[\bm{F}_{i}^{(L+1)}\circ\bm{F}_{i}^{(L)}
\circ \cdots \circ \bm{F}_{i}^{(1)}\bigr](x_i), 
~\hbox{with}~\bm{F}_{i}^{(l)}(\bm{X})= \bm{\eth}^{(l)} \bigl(\bm{W}_{i}^{(l)}\bm{X} 
+ \bm{b}_{i}^{(l)}\bigr),
\end{equation}
where $\bm{\eth}^{(l)}(\bm{X})$ is the activation function, $\bm{\theta}_i=\{\bm{W}_i^{(l)},\bm{b}_i^{(l)}, l=1, \cdots, L\}$ are parameters, and $L$ denotes the number of hidden layers. 
Denote the number of neurons for each sub-neural network by $[N_{0},N_1,N_2$, $\cdots$, $N_L$, $N_{L+1}]$ 
with input $N_0=1$ and output $N_{L+1}=p$. One has $\bm{\eth}^{(l)}(\bm{X})=\left ( \eth(X_1), \eth(X_2), \dots, \eth(X_{N_l}) \right )^{\top} $ with $ \bm{W}_{i}^{(l)}\in \mathbb{R}^{N_l\times N_{l-1}}$ and $\bm{b}_{i}^{(l)}\in\mathbb{R}^{N_l}$. In the networks,
$\phi_{i,t}$ is normalized in $L^2(\Omega)$ norm, and each sub-neural network $\bm{\phi} _i: \Omega_i \to \mathbb{R}^p$ with parameters $\bm{\theta}_i$ maps a $1$-dimensional input $x_i$ to a $p$-dimensional output through Eq.~\eqref{eq::FNN_basis}. The TNN constructs a rank-$p$ subspace $V_{\mathcal{T}}^{p} \subset H^m(\Omega)$ with the tensor 
expression \cite{wang2024solving}
\begin{equation}
V_{\mathcal{T}}^{p} = \operatorname{span}\left\{ \Phi_{t}(\bm{x};\bm{\theta})= \prod_{i=1}^d \phi_{i,t}(x_i; \bm{\theta}_i) \right\}_{t=1}^p \subset H^m(\Omega_1) \otimes \cdots \otimes H^m(\Omega_d).
\label{eq::basis_function_space}
\end{equation}
The TNN finds an approximation $\Psi_{\mathcal{T}}^p(\bm{x};\bm{\theta})$ in space $V_{\mathcal{T}}^p$ to the solution $\Psi(\bm{x})$ such that   
\begin{equation}\label{eq::TNN_structure_explain}
\Psi(\bm{x})\approx\Psi_{\mathcal{T}}^{p}(\bm{x};\bm{\theta})=\sum_{t=1}^{p}\alpha_t \cdot\prod_{i=1}^{d}\phi_{i,t}(x_{i};\bm{\theta}_{i}),    
\end{equation}
where coefficients $\alpha_t$ can be determined under the Galerkin framework \cite{SIRIGNANO20181339}. 
By Eq.~\eqref{eq::TNN_structure_explain}, the high-dimensional integral of $\Psi(\bm{x})$ on $\Omega$ can be handled via a separation of variables
\begin{equation}\label{eq::Integral_psi}
\begin{aligned}
\int_{\Omega}\Psi(\bm{x})\mathrm{d}\bm{x}\approx\sum_{t=1}^{p}\alpha_t\cdot\left[\prod_{i=1}^{d}\int_{\Omega_i}\phi_{i,t}(x_i;\bm{\theta}_i)\mathrm{d}x_i\right].
\end{aligned}
\end{equation}
It follows directly that if the integrand consists of the product of $\Psi(\bm{x})$ and other factors that also have a low-rank tensor product structure, then the corresponding high-dimensional integral can be reduced similarly to a sum of products of one-dimensional integrals. Using Gauss-Legendre quadrature to compute these one-dimensional integrals, the TNN architecture provides high-order approaches for high-dimensional integrals.  
\begin{figure}[!ht] 
\centering
\includegraphics[width=0.9\textwidth]{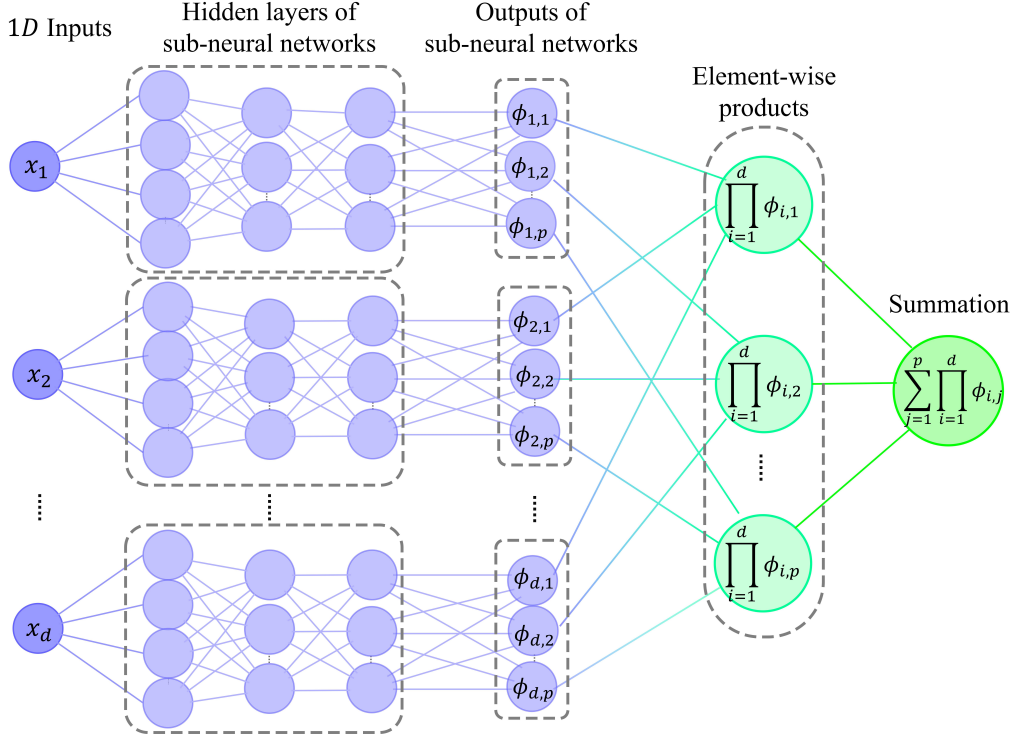}
\caption{The TNN architecture for approximating solutions to high-dimensional equations. It consists of $d$ independent feedforward sub-neural networks mapping $\mathbb{R}$ to $\mathbb{R}^p$. The outputs are multiplied to form $p$ tensor-product basis functions that are linearly combined to approximate the solution.}\label{fig:schema}
\end{figure}

It is a well-established result that the sub-neural network \eqref{eq::FNN_basis} can approximate any continuous function on a compact set \cite{Leshno1993, Ellacott1994}. The aforementioned isomorphism between Sobolev spaces implies that the TNN can also approximate well any continuous function \cite[Theorem 2.1]{Wang2024Tensor} with the corresponding rank. In contrast, the accuracy 
depends on how well the subspace $V_{\mathcal{T}}^{p}$ approximates the target solution in space $H^m(\Omega)$. 
To show the error bound,  let \(\{\psi_{\mathbf{k}}(x)\}\) be the \(d\)-dimensional Fourier basis on \(\Omega^{d}\), with \(\bm{k}=(k_1,\dots,k_d)\in\mathbb{Z}^d\). Let $t$ and $s\in\mathbb{R}$. Assume that $\Psi(\bm{x}) \in H_{\text{mix}}^{t, s}(\Omega^d)$ defined by \cite{griebel2007sparse, yserentant2004regularity}
\begin{equation}
H_{\mathrm{mix}}^{t,s}(\Omega^d) := \left\{ u(\bm{x}) = \sum_{\bm{k} \in \mathbb{Z}^d} c_{\bm{k}} \psi_{\bm{k}}(\bm{x})\Big| \|u\|_{H_{\mathrm{mix}}^{t,s}(\Omega^d)}< \infty \right\},
\end{equation}
where the norm is given by
$\|u\|_{H_{\mathrm{mix}}^{t,s}(\Omega^d)} = \left( \sum_{\bm{k} \in \mathbb{Z}^d} \lambda_1(\bm{k})^{2t} \cdot \lambda_2(\bm{k})^{2s} \cdot |c_{\bm{k}}|^2 \right)^{1/2},$
with $\lambda_1(\bm{k}) = \prod_{i=1}^d (1 + |k_i|)$ and $\lambda_2(\bm{k}) = 1 + \sum_{i=1}^d |k_i|.$
When $t>0$ and $m>s$, there exists a rank-$p$ TNN approximation $\Psi_{\mathcal{T}}^{p}(\bm{x};\bm{\theta})$ defined by \eqref{eq::TNN_structure_explain} such that the following error bound holds \cite[Theorem 2.2]{Wang2024Tensor}
\begin{align}\label{eq::TNN_approx}
\|\Psi(\bm{x})-\Psi_{\mathcal{T}}^{p}(\bm{x};\bm{\theta})\|_{H^m(\Omega^d)}\leq C(d)\cdot p^{-(s-m+t)}\cdot\|\Psi(\bm{x})\|_{H_{\mathrm{mix}}^{t,s}(\Omega^d)},
\end{align}
where $C(d) \le c \cdot d^2 \cdot 0.97515^d$ with $c$ being independent of $d$.
It is remarked that the convergence result (\ref{eq::TNN_approx}) requires high regularity of the solution. In spite that the 
wavefunctions of the Schr\"{o}dinger equation may not satisfy high regularity \cite{yserentant2004regularity}, our results show that the TNN solution achieves high accuracy with the increase of the basis size.

Specifically for the high-dimensional Schr\"{o}dinger equation, since the configuration of electrons can be restored as
\begin{equation}\label{eq::electrons}
\bm{r}=(\bm{r}_1,\cdots,\bm{r}_N)=(x_1,y_1,z_1,x_2,y_2,z_2,\cdots,x_N,y_N,z_N)\in \mathbb{R}^{3N},
\end{equation}
the wavefunction $\Psi(\bm{r})$ has a TNN approximation $\Psi_{\mathcal{T}}^{p}(\bm{r};\bm{\theta})$ with the form of Eq.~\eqref{eq::TNN_structure_explain} and dimension $d=3N$. The only remaining question is how to construct the loss function for optimization. Evidently, if one were to use only the energy functional $\mathscr{E}[\Psi]$ derived from the Rayleigh principle in Eq.~\eqref{eq::Rayleigh}, the optimization could venture into an unphysical region-specifically, the complement of $V_{\mathcal{A}}$ where the Pauli exclusion principle mentioned in Eq.~\eqref{eq::Pauli_exclusion} is violated-thus failing to obtain the true ground-state energy and the corresponding wavefunction. 

Following the SHE-TNN method for the Schr\"{o}dinger equation \cite{202209v2}, we incorporate the Pauli exclusion principle into the loss function as a penalty term. Since $T_{ij}$ is a self-adjoint unitary operator satisfying $T_{ij}^2=T_{ij}T_{ij}^{*}=id$ with eigenvalues $\pm 1$, the minimization of its Rayleigh quotient acts as a spectral constraint that targets the antisymmetric eigenstate corresponding to the lowest eigenvalue $-1$. This ensures that the low-rank tensor product approximation of the wavefunction $\Psi_{\mathcal{T}}^{p}(\bm{r};\bm{\theta})\in V_{\mathcal{T}}^{p}\cap V_{\mathcal{A}}$ obtained from the optimization is physically compliant. Denote $N_{\uparrow}=\lceil N/2\rceil$ and $N_{\downarrow}=N-N_{\uparrow}$ as the number of electrons with spin-up and spin-down states, respectively. Set $\tau_{\uparrow}:[N_{\uparrow}]\rightarrow[N]$ and $\tau_{\downarrow}:[N_{\downarrow}]\rightarrow[N]$ as the corresponding index mappings to the global indices of electrons. Then the solution $\Psi_{\mathcal{T}}^{p}(\bm{r}; \bm{\theta}^{*})$ of the ground state wavefunction is derived from 
\begin{equation}
\begin{aligned}
\Psi_{\mathcal{T}}^{p}(\bm{r};\bm{\theta}^*) = \mathop{\text{argmin}}_{\Psi(\bm{r};\theta)\in V_{\mathcal{T}}^{p}} \mathcal{L}[\Psi]   
\end{aligned}
\label{eq:pauli_constraint}
\end{equation}
with the following loss function
\begin{equation}\label{eq::loss}
\mathcal{L}[\Psi]=\frac{\langle\Psi|\hat{H}|\Psi\rangle}{\langle\Psi|\Psi\rangle} 
+ \sum_{i<j}^{N_\uparrow}\lambda_{ij}^\uparrow\frac{\langle T_{\tau_{\uparrow}(i)\tau_{\uparrow}(j)}\Psi|\Psi\rangle}{\langle\Psi|\Psi\rangle} 
+\sum_{i<j}^{N_\downarrow}\lambda_{ij}^\downarrow\frac{\langle T_{\tau_{\downarrow}(i)\tau_{\downarrow}(j)}\Psi|\Psi\rangle}{\langle\Psi|\Psi\rangle},  
\end{equation}
where $\lambda_{ij}^{\uparrow}$ and $\lambda_{ij}^{\downarrow}$ are penalty parameters for spin-up and spin-down pairs, respectively. By adjusting these hyperparameters, we can explicitly control the strictness of the antisymmetry constraint. Similar strategies are widely adopted in the loss function design for parametric PDEs, where hyperparameters balance the optimization focus between data fidelity and differential operator residuals~\cite{E2018, Raissi2019}. For notational simplicity, we denote the wavefunction $\Psi_{\mathcal{T}}^{p}(\bm{r};\bm{\theta})$ obtained at each step of the TNN process as $\Psi$, and the mass term $\langle\Psi|\Psi\rangle$ is calculated by 
\begin{equation}
\begin{aligned}
\langle\Psi|\Psi\rangle&=\sum_{t_1,t_2=1}^{p}\alpha_{t_1}\alpha_{t_2}\int_{\Omega}\prod_{i=1}^{N}\phi_{i,t_1}^x(x_{i})\phi_{i,t_1}^y(y_{i})\phi_{i,t_1}^z(z_{i})\phi_{i,t_2}^x(x_{i})\phi_{i,t_2}^y(y_{i})\phi_{i,t_2}^z(z_{i})\mathrm{d}\bm{r}\\      &=\sum_{t_1,t_2=1}^{p}\alpha_{t_1}\alpha_{t_2}\prod_{i=1}^{N}\prod_{\eta\in \{x,y,z\}}\int_{\Omega_0}\phi_{i,t_1}^{\eta}(\eta_i)\phi_{i,t_2}^{\eta}(\eta_i)\mathrm{d}\eta_i,
\end{aligned}
\label{eq::Mass}
\end{equation}
where the $3N$ basis are distinguished by their coordinates in the superscript $\eta\in \{x,y,z\}$. Compared to the high-dimensional integration, the computational complexity dramatically decreases from $\mathcal{O}(K^{3N})$ to $\mathcal{O}(3p^2NK)$, where $K$ denotes the number of proxy points applied on each dimension. The exchange inner product $\langle T_{ij}\Psi|\Psi\rangle$ can be calculated analogously as Eq.~\eqref{eq::Mass} since all coordinates are still separable in the integral. The term containing the Hamiltonian mentioned in Eq.~\eqref{Hamiltonian_Operator} can be further decomposed as
\begin{equation}
\begin{aligned}
\frac{\langle\Psi|\hat{H}|\Psi\rangle}{\langle\Psi|\Psi\rangle} 
&= \sum_{i=1}^{N}\frac{\int_{\Omega}|\nabla_i\Psi|^2\mathrm{d}\bm{r}}{2 \int_{\Omega}|\Psi|^2\mathrm{d}\bm{r}}+\sum_{k<l}^{M}\frac{Q_{k}Q_{l}}{|\bm{R}_{k}-\bm{R}_{l}|} 
- \sum_{i=1}^{N}\sum_{k=1}^{M}\frac{Q_k\int_{\Omega}\frac{|\Psi|^2}{|\bm{r}_i-\bm{R}_k|}\mathrm{d}\bm{r}}{\int_{\Omega}|\Psi|^2\mathrm{d}\bm{r}} + \sum_{i<j}^{N}\frac{\int_{\Omega}\frac{|\Psi|^2}{|\bm{r}_i-\bm{r}_j|}\mathrm{d}\bm{r}}{\int_{\Omega}|\Psi|^2\mathrm{d}\bm{r}}   \\ 
&\triangleq E_{\text{K}}+E_{ii}- E_{ie} + E_{ee},
\end{aligned}
\label{eq::Energydef}
\end{equation}
where the four parts represent the kinetic term, the ion-ion interaction term, the ion-electron term, as well as the electron-electron term, respectively. For the low-rank tensor structure of Eq.~\eqref{eq::TNN_structure_explain}, the gradient operator can be conveniently applied to the basis functions of each separated dimension. This allows for the rapid computation of $E_{\text{K}}$ using one-dimensional integrals, in a manner fundamentally consistent with Eq.~\eqref{eq::Mass}. Under the Born-Oppenheimer approximation, the ion-ion interaction term $E_{ii}$ is a constant value that is independent of the electronic distribution and then its exact value can be precomputed. For the electron-related interactions in Eq.~\eqref{eq::Energydef}, $E_{ie}$ and $E_{ee}$, the presence of the Coulomb term in the integral prevents the corresponding three- and six-dimensional integrals from being decomposed into a product of one-dimensional integrals. This becomes the computational bottleneck for solving the Schr\"{o}dinger equation within the TNN framework. In what follows, we describe how the SOG decomposition is introduced to address the problem.

\subsection{Tensor representation via the SOG decomposition}\label{sec::TRSOG}
The key step of the TNN framework is to reduce all high-dimensional integrals to a low-rank sum of products of one-dimensional integrals. For the purpose, the Coulomb interaction in the Schr\"{o}dinger equation must also be treated in a way that facilitates computation using the tensor structure. The SOG decomposition is useful to achieve this goal. The SOG has outstanding performance in field of molecular dynamics simulations \cite{predescu2020,RBSOG,Gao2024Fast,chen2025random,ji2025machine}, and the spatial separability of a Gaussian makes it highly compatible with the TNN framework.  The Coulomb kernel has the integral identity
\begin{equation}\label{eq::SOG_integral}
\frac{1}{r} = \frac{1}{\sqrt{\pi}} \int_{-\infty}^{\infty} e^{-e^t r^2 + \frac{1}{2} t} \mathrm{d}t=\frac{2}{\sqrt{\pi}}\int_{0}^{\infty}e^{-x^2r^2}\mathrm{d}x.
\end{equation} 
By a geometrically spaced quadrature with nodes $x_\ell=1/(\sqrt{2}b^{\ell}\sigma)$ for $\ell\in \mathbb{Z}$ and $b>1$, an approximation with a SOG series holds,
\begin{equation}
\frac{1}{r} \approx \frac{2\ln b}{\sqrt{2\pi}\sigma} \sum_{\ell=-\infty}^{\infty} \frac{1}{b^\ell} \exp\!\left[-\frac{1}{2}\left(\frac{r}{b^\ell\sigma}\right)^2\right].
\label{eq:SOGof1/r}
\end{equation}
This is the bilateral series approximation (BSA) \cite{beylkin2005, beylkin2010} for the Coulomb potential.
The relative error of the BSA is asymptotically bounded by \cite{predescu2020}
\begin{align}
\left|1 - \frac{2r\ln b}{\sqrt{2\pi}\sigma} \sum_{\ell=-\infty}^{\infty} \frac{1}{b^\ell} \exp\!\left[-\frac{1}{2}\left(\frac{r}{b^\ell\sigma}\right)^2\right]\right|
\lesssim 2\sqrt{2}\exp\!\left(-\frac{\pi^2}{2\ln b}\right).\
\label{eq::sog_error}
\end{align}
If one truncates the infinite series \eqref{eq:SOGof1/r} such that the summation is 
from $-M_2$ to $M_1$, the effective approximation range is 
roughly $\left[\sqrt{2}b^{-M_2}\sigma, \sqrt{2}b^{M_1}\sigma\right]$ \cite{LIANG2025101759}.
Based on these results, a corresponding set of SOG parameters can be determined for any given accuracy.

By the SOG decomposition, we first address the high-dimensional integral arising from the ion-electron interaction term $E_{ie}$. Based on the TNN structure, the integral for the interaction between the $i$th electron and the $k$th ion can be approximated by
\begin{equation}\label{eq::E_ie}
\int_{\Omega}\frac{|\Psi|^2}{|\bm{r}_i-\bm{R}_k|}\mathrm{d}\bm{r}=\sum_{t_1,t_2=1}^{p}\alpha_{t_1}\alpha_{t_2}
\int_{\Omega_0^{\otimes 3}}\frac{\phi_{i,t_1}(\bm{r}_{i})
\phi_{i,t_2}(\bm{r}_{i})}{|\bm{r}_i-\bm{R}_k|}
\mathrm{d}\bm{r}_i\prod_{j\neq i}\prod_{\eta\in \{x,y,z\}}
\int_{\Omega_0}\phi_{j,t_1}^{\eta}(\eta_j)\phi_{j,t_2}^{\eta}(\eta_j)\mathrm{d}\eta_j, 
\end{equation}
where $\phi_{i,t}(\bm{r}_{i})=\phi_{i,t}^x(x_{i})\phi_{i,t}^y(y_{i})\phi_{i,t}^z(z_{i})$ is also separable. The quantity that needs to be handled is the three-dimensional integral for each electron $\bm{r}_i$. We employ the SOG approximation as Eq.~\eqref{eq:SOGof1/r} with truncation $(M_1,M_2)$ for the Coulomb interaction such that

\begin{equation}
\frac{1}{|\bm{r}_i-\bm{R}_k|} \thickapprox\sum_{\ell=-M_2}^{M_1} w_{\ell}G_{\ell}(x_i - X_k)G_{\ell}(y_i - Y_k)G_{\ell}(z_i - Z_k),
\label{eq::ie_SOG}
\end{equation}
where  $w_{\ell}=4\pi s_{\ell}^2\ln{b}$ is the weight, and $G_{ \ell}(x) = (1/\sqrt{2\pi} s_{\ell}) \exp\left( -x^2/2s_{\ell}^2\right)$
is the Gaussian function of bandwidth $s_{\ell}=b^{\ell}\sigma$. Substituting Eq.~\eqref{eq::ie_SOG} into Eq.~\eqref{eq::E_ie} leads us to

\begin{equation}\label{eq::ie_split}
\begin{aligned}
\int_{\Omega_0^{\otimes 3}}\frac{\phi_{i,t_1}(\bm{r}_{i})\phi_{i,t_2}(\bm{r}_{i})}{|\bm{r}_i-\bm{R}_k|}\mathrm{d}\bm{r}_i&\approx\sum_{\ell=-M_2}^{M_1}w_{\ell}\prod_{\eta\in\{x,y,z\}}\int_{\Omega_0}\phi_{i,t_1}^{\eta}(\eta_{i})\phi_{i,t_2}^{\eta}(\eta_{i})G_\ell(\eta_i-R_k^{\eta})\mathrm{d}\eta_i,
\end{aligned}
\end{equation}
where, for convenience, $R_k^{\eta}$ represents the $\eta$-component of $\bm{R}_k$, namely, $R_k^{x}=X_k$, $R_k^{y}=Y_k$ and $R_k^{z}=Z_k$. Thus, the contribution from the ion-electron term $E_{ie}$ can be obtained by computing a product of one-dimensional integrals, without requiring any high-dimensional integration. In a similar manner for the electron-electron interaction $E_{ee}$, one reduces the six-dimensional integral to the product of two-dimensional integrals,
\begin{equation}
\begin{aligned}\label{eq::ee_split}
&\int_{\Omega_0^{\otimes 6}}\frac{\phi_{i,t_1}(\bm{r}_{i})\phi_{j,t_1}(\bm{r}_{j})\phi_{i,t_2}(\bm{r}_{i})\phi_{j,t_2}(\bm{r}_{j})}{|\bm{r}_i-\bm{r}_j|}\mathrm{d}\bm{r}_i\mathrm{d}\bm{r}_j\\
\approx&\sum_{\ell=-M_2}^{M_1}w_{\ell}\prod_{\eta\in\{x,y,z\}}\int_{\Omega_0^{\otimes 2}}\phi_{i,t_1}^{\eta}(\eta_{i})\phi_{j,t_1}^{\eta}(\eta_{j})\phi_{i,t_2}^{\eta}(\eta_{i})\phi_{j,t_2}^{\eta}(\eta_{j})G_\ell(\eta_i-\eta_j)\mathrm{d}\eta_i\mathrm{d}\eta_j.
\end{aligned}    
\end{equation}
Thus, the SOG-TNN method has simplified the problem to the computation of only one- and two-dimensional integrals. In the following section, we develop fast algorithms to further 
speed up computing the two-dimensional integrals in the electron-electron interactions such that one has to calculate only a small portion of two-dimensional integrals. This approach will be the subject of Section~\ref{sec::Range_split}.

\section{Range-splitting acceleration}
\label{sec::Range_split}
In this section, we develop accelerating techniques for electron-electron interactions by exploiting the bandwidth-dependent intrinsic low-rank properties of Gaussian functions. Similar to Gao {\it et al.} \cite{Gao2024Fast}, we divide the indices of all the Gaussians into three sets. Depending on the values of $s_\ell$, these index sets are denoted in increasing order by $\mathcal{M}_{\text{short}}$, $\mathcal{M}_{\text{mid}}$ and $\mathcal{M}_{\text{long}}$, which represent the short-range, mid-range, and long-range components, respectively. This partitioning criterion can be explicitly determined by the  complexity analysis under the accuracy requirement.

\subsection{Asymptotic calculation for short-range components} 
For each $\ell\in \mathcal{M}_{\text{short}}$, the Gaussian distribution concentrates its mass near the origin. As $s_\ell\rightarrow 0$, by the moment expansion, one has the following asymptotic series for any smooth function $f(x)$,
\begin{equation}
\label{moments}
\int_{\Omega_0}f(x)G_{\ell}(x)dx\sim \sum_{k =0}^{+\infty}\frac{f^{(2k )}(0)}{(2k )!!}s_{\ell}^{2k }. 
\end{equation}
Employing this expansion to the integral in Eq. \eqref{eq::ee_split} and taking the leading term, one obtains,  
\begin{equation}
    \label{eq::short_approx_integral}
    \begin{aligned}
    &\int_{\Omega_0^{\otimes 2}}\phi_{i,t_1}^{\eta}(\eta_{i})\phi_{j,t_1}^{\eta}(\eta_{j})\phi_{i,t_2}^{\eta}(\eta_{i})\phi_{j,t_2}^{\eta}(\eta_{j})G_\ell(\eta_i-\eta_j)\mathrm{d}\eta_i\mathrm{d}\eta_j\\
    \approx&\int_{\Omega_0}\phi_{i,t_1}^{\eta}(\eta_i)\phi_{j,t_1}^{\eta}(\eta_i)\phi_{i,t_2}^{\eta}(\eta_i)\phi_{j,t_2}^{\eta}(\eta_i)\mathrm{d}\eta_i\triangleq\mathcal{I}_{t_1,t_2,i,j}^{\eta,\text{short}}.
    \end{aligned}
\end{equation}
By this approximation, one represents the short-range part of the electron-electron energy by a series of the product of one-dimensional integrals, expressed by, 
\begin{equation}
    \label{eq::Short_contribution}
E_{ee}^{\text{short}}=\left(\sum_{\ell\in\mathcal{M}_{\text{short}}}w_{\ell}\right)\sum_{t_1,t_2=1}^{p}\alpha_{t_1}\alpha_{t_2}\sum_{i<j}^{N}\left(\mathcal{J}_{t_1,t_2,i,j}\cdot\prod_{\eta\in\{x,y,z\}}\mathcal{I}_{t_1,t_2,i,j}^{\eta,\text{short}}\right),
\end{equation}
where 
\begin{equation}
    \label{eq::J}
    \mathcal{J}_{t_1,t_2,i,j}:=\prod_{k\neq i,j}^{N}\prod_{\eta\in\{x,y,z\}}\int_{\Omega_0}\phi_{j,t_1}^{\eta}(\eta_k)\phi_{j,t_2}^{\eta}(\eta_k)\mathrm{d}\eta_k
\end{equation}
is the remaining tensor-product term except for the electron-electron interaction. One can see that Eq.~\eqref{eq::Short_contribution} is the result of the linear superposition of all the Gaussians at the zero-bandwidth limit. It allows to extend the summation to include all negative indices (effectively setting $M_2=\infty$) such that the approximation error of $E_{ee}^{\text{short}}$ is only from the asymptotic expansion. This strategy for the short-range component can likewise be applied to the ion-electron interaction described in Eq.~\eqref{eq::ie_split}. Thus, the SOG-TNN method can completely resolve the singularity issue of the Coulomb kernel, which is crucial to achieving stable and high-precision computation for high-dimensional Schr\"{o}dinger equation.

\subsection{Chebyshev tensorization for long-range components}
For Gaussians with $\ell\in \mathcal{M}_{\text{long}}$, they are smooth and vary slowly within domain $\Omega_0$. Inspired by the idea of the fast Gauss transform \cite{greengard1991fast,chebyshev,greengard2024new}, their intrinsic low-rank nature can be effectively captured by an expansion in orthogonal Chebyshev polynomials. Specifically, one has the following Chebyshev expansion
\begin{equation}\label{eq:chebyshev_expansion}
G_{\ell}\left(x-y\right)=\sum_{m,n=0}^{+\infty} C_{m,n}^{r_c}(\ell) T_m\left(\frac{x}{r_c}\right)T_n\left(\frac{y}{r_c}\right),\quad x,y\in \Omega_0, 
\end{equation}
where $T_n(x)=\cos(n\arccos x)$ denotes the $n$th Chebyshev polynomial, and the coefficients can be precomputed by
\begin{equation}\label{eq:Chebycoeff}
    C_{m,n}^{r_c}(\ell)=\frac{\gamma_m\gamma_n} {\pi^2}\int_{[-1,1]^2}\frac{r_c}{\sqrt{2\pi}s_{\ell}}\exp\left(-\frac{r_c^2(x-y)^2}{2s_{\ell}^{2}}\right)T_m(x)T_n(y)\omega(x)\omega(y)\mathrm{d}x\mathrm{d}y 
\end{equation}
with the weight function $\omega(x)=1/\sqrt{1-x^2}$, $\gamma_0=1$ and $\gamma_k=2$ for $k\ge 1$. If we truncate the infinite sum of Eq.~\eqref{eq:chebyshev_expansion} with $m+n\le S$ for $S\in \mathbb{N}^{+}$, the asymptotic error bound for $x,y\in \Omega_0$ can be estimated by \cite{chebyshev}
\begin{equation}\label{eq:cheb_error}
\left|G_{\ell}\left(x-y\right)- \sum_{|m+n|\le S} C_{m,n}^{r_c}(\ell) T_m\left(\frac{x}{r_c}\right)T_n\left(\frac{y}{r_c}\right)\right| \lesssim \frac{r_c^S}{\sqrt{2}(2s_{\ell} S^{1/2})^{S}},\  \text{as}\  S\rightarrow+\infty,
\end{equation}
which converges factorially fast with respect to $S$ for $\ell \in \mathcal{M}_{\text{long}}$. The Chebyshev expansion imparts a low-rank sum-of-tensor-product structure to the Gaussian function within the two-dimensional integral of Eq.~\eqref{eq::ee_split}. As a result, the long-range component of the electron-electron interaction can be efficiently evaluated, since it is reduced to a series of one-dimensional integrals,
\begin{equation}
    \label{eq::Long_contribution}
E_{ee}^{\text{long}}=\sum_{\ell\in\mathcal{M}_{\text{long}}}w_{\ell}\sum_{t_1,t_2=1}^{p}\alpha_{t_1}\alpha_{t_2}\sum_{i<j}^{N}\left[\mathcal{J}_{t_1,t_2,i,j}\cdot\prod_{\eta\in\{x,y,z\}}\mathcal{I}_{t_1,t_2,i,j}^{\eta,\ell,\text{long}}\right],
\end{equation}
where
\begin{equation}\label{eq::long_2Dintegral}
\begin{aligned}
&\mathcal{I}_{t_1,t_2,i,j}^{\eta,\ell,\text{long}}\\
&:=\sum_{|m+n|\le S}C_{m,n}^{r_c}(\ell)\int_{\Omega_0}\phi_{i,t_1}^{\eta}(\eta_{i})\phi_{i,t_2}^{\eta}(\eta_{i})T_m\left(\frac{\eta_i}{r_c}\right)\mathrm{d}\eta_i\cdot\int_{\Omega_0}\phi_{j,t_1}^{\eta}(\eta_{j})\phi_{j,t_2}^{\eta}(\eta_{j})T_n\left(\frac{\eta_j}{r_c}\right)\mathrm{d}\eta_j
\end{aligned}
\end{equation}
becomes rank-$\tbinom{S+1}{2}$ product of two one-dimensional integrals, which avoids the computation of the original two-dimensional integral.  

\subsection{Model reduction for mid-range components} Consider how to compute those Gaussians with $\ell\in \mathcal{M}_{\text{mid}}$. The moderate bandwidth of these Gaussian functions presents a dilemma for efficient approximation. On one hand, treating them asymptotically introduces a significant approximation error. On the other hand, a Chebyshev expansion would require a large truncation number, making the computational cost prohibitively high. Although the number of Gaussians in this mid-range set $\mathcal{M}_{\text{mid}}$ is typically small, the direct evaluation of the two-dimensional numerical integrals in Eq.~\eqref{eq::ee_split} for these terms still constitutes the primary computational bottleneck. Fortunately, the complex structure of these mid-range interactions can be well settled by model reduction based on the singular value decomposition (SVD).

Consider the two-dimensional integral in Eq.~\eqref{eq::ee_split}. By applying the Gauss-Legendre quadrature to each dimension, one has
\begin{equation}
    \label{eq::mid_num_integral}
    \begin{aligned}
        \mathcal{I}_{t_1,t_2,i,j}^{\eta,\ell,\text{mid}}:=&\int_{\Omega_0^{\otimes 2}}\phi_{i,t_1}^{\eta}(\eta_{i})\phi_{j,t_1}^{\eta}(\eta_{j})\phi_{i,t_2}^{\eta}(\eta_{i})\phi_{j,t_2}^{\eta}(\eta_{j})G_\ell(\eta_i-\eta_j)\mathrm{d}\eta_i\mathrm{d}\eta_j\\
    \approx&\sum_{k,l=1}^{K}(\nu_k\phi_{i,t_1}^{\eta}(\hat{\eta}_k)\phi_{i,t_2}^{\eta}(\hat{\eta}_k))\cdot G_\ell(\hat{\eta}_k-\hat{\eta}_l)\cdot (\nu_l\phi_{j,t_1}^{\eta}(\hat{\eta}_l)\phi_{j,t_2}^{\eta}(\hat{\eta}_l))\\
    =&(\bm{\nu}\circ\bm{\phi}_{i,t_1}^{\eta}\circ\bm{\phi}_{i,t_2}^{\eta})^\top\bm{G}_{\ell}(\bm{\nu}\circ\bm{\phi}_{j,t_1}^{\eta}\circ\bm{\phi}_{j,t_2}^{\eta}),\\
    \end{aligned}
\end{equation}
where $\bm{\hat{\eta}}=(\hat{\eta}_1,\cdots,\hat{\eta}_K)$ denote the quadrature nodes for each dimension, $\bm{\nu}=(\nu_1,\cdots,\nu_K)$ are the weights,  $\bm{\phi}_{i,t}^\eta=(\phi_{i,t}^\eta(\hat{\eta}_1),\cdots,\phi_{i,t}^\eta(\hat{\eta}_K))$ are the corresponding TNN values at the nodes, and ``$\circ$" denotes the Hadamard product. The kernel matrix $\bm{G}_{\ell}\in\mathbb{R}^{K\times K}$ with the $(k,l)$-element being $G_{\ell}(\hat{\eta}_k-\hat{\eta}_l)$. One applies the SVD such that  $\bm{G}_{\ell}=\boldsymbol{U}_\ell \boldsymbol{\Sigma}_\ell \boldsymbol{V}_\ell^\top$,
where the diagonal matrix $\boldsymbol{\Sigma}_\ell = \mathrm{diag}\{\lambda_1^{\ell}, \cdots, \lambda_K^{\ell}\}$  stores its singular values in a descending order, and $\bm{U}_\ell$ and $\bm{V}_\ell\in \mathbb{R}^{K\times K}$ are orthogonal. The model reduction then is applied for truncating $\bm{\Sigma}_{\ell}$ into $\widetilde{\bm{\Sigma}}_\ell=\mathrm{diag}\{\lambda_1^{\ell}, \cdots, \lambda_{r_\ell}^{\ell}\}$, retaining only those $\lambda_i^{\ell}$ that greater than the required accuracy. The two orthogonal matrices become $\widetilde{\bm{U}}_\ell$ and $\widetilde{\bm{V}}_\ell\in \mathbb{R}^{K\times r_\ell}$.  Eq.~\eqref{eq::mid_num_integral} is then approximated by
\begin{equation}
    \label{eq::mid_2Dintegral}
\mathcal{I}_{t_1,t_2,i,j}^{\eta,\ell,\text{mid}}\approx (\widetilde{\bm{U}}_\ell^\top\bm{\nu}\circ\bm{\phi}_{i,t_1}^{\eta}\circ\bm{\phi}_{i,t_2}^{\eta})^{\top}\widetilde{\bm{\Sigma}}_\ell(\widetilde{\bm{V}}_\ell^{\top}\bm{\nu}\circ\bm{\phi}_{j,t_1}^{\eta}\circ\bm{\phi}_{j,t_2}^{\eta}).
\end{equation}
In Figure~\ref{fig::SVD}, we show singular value distribution for several typical kernel matrices $\bm{G}_{\ell}$ for different bandwidth $s_\ell$. The rapid decay of these singular values demonstrates the effectiveness of SVD as a model reduction method for the mid-range Gaussian components.

\begin{figure}[!ht] 
    \centering    \includegraphics[width=0.9\textwidth]{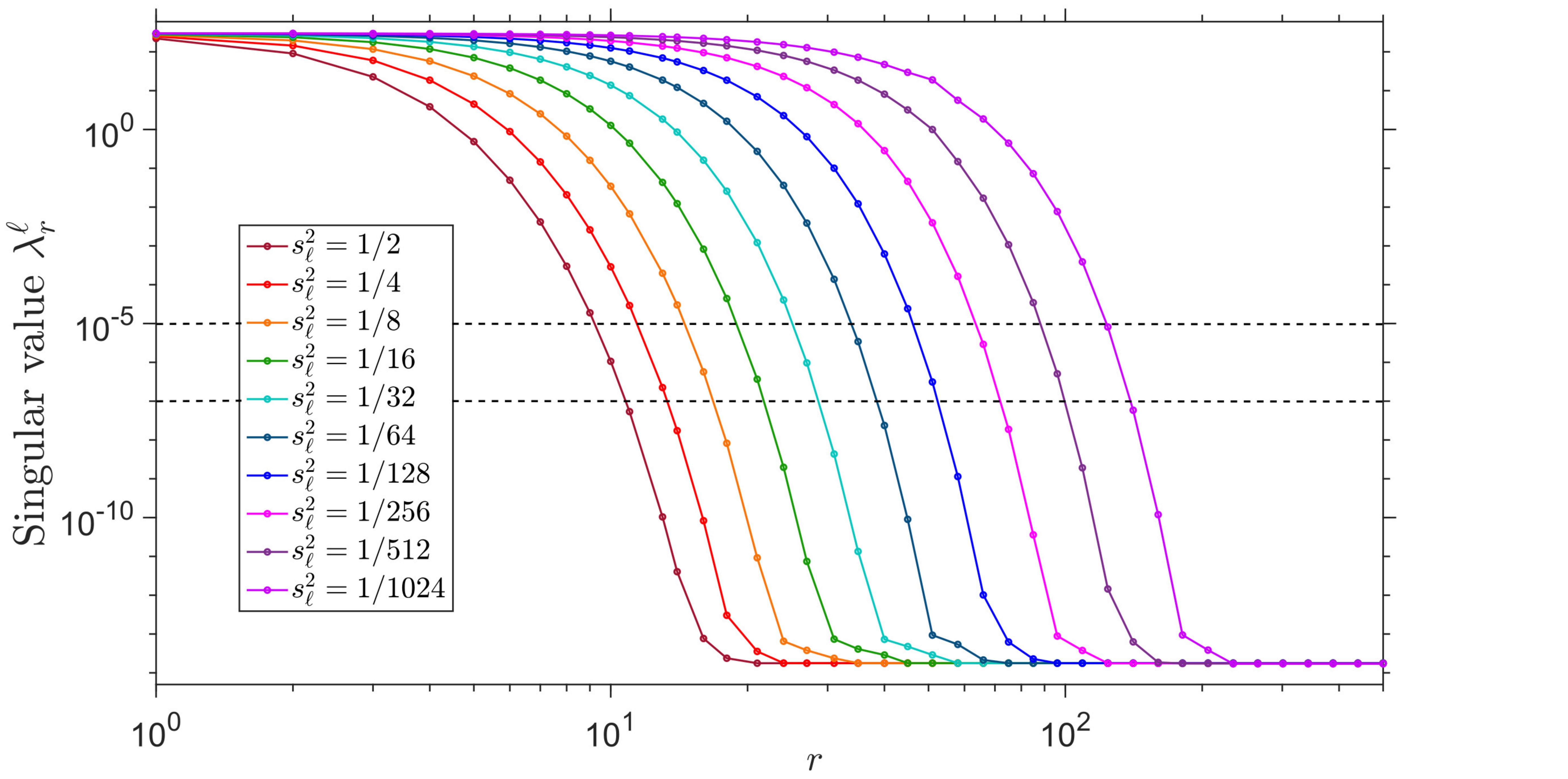}
    \caption{Singular value distribution of kernel matrices $\bm{G}_{\ell}$ consists of $600$ proxy points chosen as composite Gauss-Legendre nodes with respect to different bandwidth $s_\ell$.}
    \label{fig::SVD}
\end{figure}

Finally, the mid-range energy component reads as follows
\begin{equation}
    \label{eq::Mid_contribution}
E_{ee}^{\text{mid}}=\sum_{\ell\in\mathcal{M}_{\text{mid}}}w_{\ell}\sum_{t_1,t_2=1}^{p}\alpha_{t_1}\alpha_{t_2}\sum_{i<j}^{N}\left[\mathcal{J}_{t_1,t_2,i,j}\cdot\prod_{\eta\in\{x,y,z\}}\mathcal{I}_{t_1,t_2,i,j}^{\eta,\ell,\text{mid}}\right].
\end{equation}

\subsection{Summary of the algorithm and complexity analysis}

Combining Eqs.~\eqref{eq::Short_contribution}, \eqref{eq::Long_contribution} and \eqref{eq::Mid_contribution} together results in a complete solution for calculating $E_{ee}$.  
The most important feature of the SOG-TNN architecture is that the integration over each spatial dimension is decoupled. The separability makes the method naturally suited for anisotropic systems, such as those where the physical domain exhibits a large aspect ratio. This outcome reveals the great potential of the SOG-TNN for solving more complex systems such as long-chain molecules. The complete algorithm for the SOG-TNN solver with range-splitting acceleration for the high-dimensional Schr\"{o}dinger equation is summarized in Algorithm~\ref{alg:algorithm_SOGTNN}. Here, parameters $\boldsymbol{\theta}$ of the neural network are optimized with antisymmetry constraints with the TNN basis function space being pushed towards the antisymmetric sector. Subsequently, Hamiltonian-based generalized eigenvalue problem  optimizes $\boldsymbol{\alpha}$ within the fixed $\boldsymbol{\theta}$-subspace. This hybrid strategy is designed as an adaptive subspace construction, theoretically supported by the framework of Galerkin neural networks~\cite{ainsworth2021galerkin}.

\begin{algorithm}[h]
\caption{The SOG-TNN solver for the Schr\"{o}dinger equation}
\label{alg:algorithm_SOGTNN}
\begin{algorithmic}[1]  

\Require Configurations of the quantum system;  initial model parameters $\bm{\theta}^{(0)}$; 
TNN architecture $\Psi(\bm{r}; \bm{\theta}^{(0)})$ 
defined by \eqref{eq::TNN_structure_explain}; cutoff $r_c$; 
error tolerance $\epsilon$.

\State Based on $\epsilon$, select the number of integral nodes $K$, 
and the TNN and SOG parameters. Select the number of training iterations $N_{\text{epoch}}$, 
parameters for spin-up and spin-down pairs $\lambda_{ij}^\uparrow$ and $\lambda_{ij}^\downarrow$, 
and  the learning rate $\gamma$.
\State Calculate the ion-ion interaction $E_{ii}$ defined in Eq.~\eqref{eq::Energydef}, 
which remains constant in the whole process.
\For {$\kappa =1:N_{\text{epoch}}$}
\State Construct the subspace $V_{\mathcal{T}}^{p,(\kappa -1)}$ with $\bm{\theta}^{(\kappa -1)}$ by Eq.~\eqref{eq::basis_function_space}, and calculate the stiffness and the mass matrices 
$\left ( \left \langle \Phi_{t_1}^{(\kappa -1)} | \hat{H} | \Phi_{t_2}^{(\kappa -1)} \right \rangle  \right ) _{p\times p}$ and $ \left (  \langle \Phi_{t_1}^{(\kappa -1)} | \Phi_{t_2}^{(\kappa -1)} \rangle  \right ) _{p\times p},$ respectively. The SOG approximation is applied on the calculation of the  stiffness matrix, where $E_{ie}$ is calculated by Eq.~\eqref{eq::E_ie}, and $E_{ee}$ is calculated by Eqs.~\eqref{eq::Short_contribution}, \eqref{eq::Long_contribution}, and \eqref{eq::Mid_contribution}.
\State Solve the generalized eigenvalue problem
$$
\left( \left\langle \Phi_{t_1}^{(\kappa -1)} \middle| \hat{H} \middle| \Phi_{t_2}^{(\kappa -1)} \right\rangle \right)_{p \times p} \bm{\alpha}^{(\kappa -1)} = \bm{E}^{(\kappa -1)} \left( \left\langle \Phi_{t_1}^{(\kappa -1)} \middle| \Phi_{t_2}^{(\kappa -1)} \right\rangle \right)_{p \times p},
$$
by transforming it into a standard eigenvalue problem using Cholesky decomposition.
\State Construct the approximate solution $\Psi_{\mathcal{T}}^{p}(\bm{x};\bm{\theta}^{(\kappa -1)})$ with $\bm{\alpha}^{(\kappa -1)}$ by Eq.~\eqref{eq::TNN_structure_explain}, then compute the loss function $\mathcal{L}[\Psi^{(\kappa -1)}]$ as Eq.~\eqref{eq::loss}.

\State Update TNN parameters from $\bm{\theta}^{(\kappa-1)}$ to $\bm{\theta}^{(\kappa)}$ using some deterministic optimization steps based on the Adam optimizer.
\EndFor{}
\Ensure The ground state eigenvalue $\mathscr{E}[\Psi](\bm{\theta}^{(N_{\text{epoch}})})$ and wavefunction $\Psi(\bm{r};\bm{\theta}^{(N_{\text{epoch}})})$. 
\end{algorithmic}
\end{algorithm}

Let us analyze the computational complexity of a single-step calculation using the proposed SOG-TNN.  In Step 4, the complexity of computing the mass matrix is $\mathcal{O}(3p^2NKT_{\text{eval}})$, where $T_{\text{eval}}$ denotes the computational complexity for the one dimensional function evaluation operations. The computation of the stiffness matrix is divided into three main components, $E_{\text{K}}$, $E_{ie}$ and $E_{ee}$. The kinetic energy term $E_{\text{K}}$ involves the integral of the inner product of the gradients of $\Psi$. Since the gradient operation preserves the tensor structure, this term can be assembled simply by computing the inner products of the differentiated one-dimensional basis functions, resulting in the same computational complexity of $\mathcal{O}(9p^2NKT_{\text{eval}})$. The ion-electron interaction contribution $E_{ie}$ is computed according to Eqs.~\eqref{eq::E_ie} and \eqref{eq::ie_split}, and has a computational complexity of
\begin{equation}\label{eq::Complexity_ie}
\mathcal{C}_{ie}=\mathcal{O}\left(3p^2MN+3|\mathcal{M}|p^2MNKT_{\text{eval}}\right),\  \mathcal{M}:=\mathcal{M}_{\text{mid}}\cup\mathcal{M}_{\text{long}}.
\end{equation} 
The electron-electron interaction contribution $E_{ee}$ is composed of the three contributions by Eqs.~\eqref{eq::Short_contribution}, \eqref{eq::Long_contribution}, and \eqref{eq::Mid_contribution}, and the total computational complexity is
\begin{equation}\label{eq::Complexity_ee}
\mathcal{C}_{ee}=\mathcal{O}\left(3p^2N^2KT_{\text{eval}}+(3|\mathcal{M}_{\text{long}}|(S+1)p^2NKT_{\text{eval}}+S^2N^2)+3\sum_{\ell\in\mathcal{M}_{\text{mid}}}p^2N^2r_\ell KT_{\text{eval}}\right),
\end{equation}
where the three terms with the big $\mathcal{O}$ correspond to the short-, long- and mid-range calculations, respectively. Specifically, for the computation of the long-range contribution in Eq.~\eqref{eq::long_2Dintegral}, the summation of the Chebyshev expansion is handled efficiently. First, the constituent one-dimensional integrals within the term $\mathcal{I}_{t_1,t_2,i,j}^{n,\ell,\text{long}}$ are pre-computed and stored. The final result is then assembled through an $O(S^2N^2)$ linear combination of these integrals, further reducing the overall computational cost for this component. Step 5 requires solving a generalized eigenvalue problem. Both the Cholesky decomposition and the solution of the subsequently transformed standard eigenvalue problem have a computational complexity of $\mathcal{O}(p^3)$ \cite{golub2013matrix}. Step 6 then involves the computation of the penalty term in Eq.~\eqref{eq::loss} to enforce the Pauli exclusion principle, which has a complexity of $\mathcal{O}(3p^2N^2KT_{\text{eval}})$.

To summarize the above analysis, the total computational complexity of the SOG-TNN method for one step is 
\begin{equation}\label{eq::total_complexity}
\mathcal{C}_{\text{total}}=\mathcal{O}\left(p^3+S^2N^2+(S|\mathcal M_{\rm long}|+N|\mathcal M_{\rm mid}|+M|\mathcal{M}|)p^2KNT_{\text{eval}}\right),  
\end{equation}
which shows that the architecture overcomes the challenge of high-dimensional integrations. Furthermore, the Gaussians are also well-suited for sparse storage, thus yielding substantial savings in memory overhead. Consequently, it is not difficult to estimate the order of the memory cost by the SOG-TNN method, which is 
$\mathcal{O}\left(pNK+p^2+K+S^2+SK+M\right).$
Moreover, the computational complexity of the SOG-TNN linearly depends on the number of Gaussians, making the method amenable to further acceleration by reducing the number of Gaussians. Model reduction techniques can be applied to construct optimized SOG expansions; see, for example, \cite{greengard2018anisotropic,kammler1977prony,beylkin2019computing,
gao2022kernel,lin2025weighted} for a review of such methods.

\subsection{Enforcement of antisymmetry constraints}

The above SOG-TNN framework incorporates the wavefunction antisymmetry as a soft constraint via an additional penalty in the loss function. Actually, it is also possible to construct an antisymmetric TNN ansatz that strictly preserves the antisymmetry. Let $\mathcal{A}:H^m(\Omega;\mathbb{R}^{n})\rightarrow H^{m}(\Omega^n)$ denote the antisymmetrization operator for any $n\in \mathbb{N}_{+}$ and $\bm{\Phi}=[\phi_1(r_1),\cdots,\phi_n(r_n)]^{\top}$, such that one has
\begin{equation}
    \label{eq::antisym_op}
    \mathcal{A}\left(\bm{\Phi}\right)=\frac{1}{n!} \begin{vmatrix}
\phi_{1}(r_1)  & \dots &\phi_{1}(r_{n}) \\
 \vdots & \ddots  & \vdots\\
\phi_{n}(r_1)  & \dots &\phi_{n}(r_{n})
\end{vmatrix}.
\end{equation}
This operator originates from the Slater determinant \cite{slater1929theory}. By acting on all electrons with the same spin, it inherently preserves the antisymmetry. Analogous to the Hartree-Fock-like approaches in literature~\cite{fock1930naherungsmethode,beylkin2008approximating}, the TNN ansatz of the wavefunction $\Psi$ can be reformulated as
\begin{equation}   \label{eq::TNN_Slater}
    \Psi\approx \Psi_{\mathcal{A}}^{p}=\sum_{t=1}^{p}\alpha_t\mathcal{A}\left(\bm{\Phi}_t^{\uparrow}\right)\mathcal{A}\left(\bm{\Phi}_t^{\downarrow}\right),
\end{equation}
where
\begin{equation}
    \label{eq::Phi_up_down} \bm{\Phi}_t^{\uparrow}=\left[(\phi_{\tau_{\uparrow}(i),t}(r_{\tau_{\uparrow}(i)}))_{i=1}^{N_{\uparrow}}\right]^{\top},\quad \hbox{and} \quad \bm{\Phi}_t^{\downarrow}=\left[(\phi_{\tau_{\downarrow}(j),t}(r_{\tau_{\downarrow}(j)}))_{j=1}^{N_{\downarrow}}\right]^{\top}
\end{equation}
denote the spin-up and spin-down TNN component vectors, respectively. Notably, this antisymmetric TNN ansatz preserves the structure of the individual sub-neural network shown in Figure~\ref{fig:schema}, with the action of $\mathcal{A}$ being applied solely at the final assembly stage. Under this enforced antisymmetric ansatz, the high-dimensional integration within SOG-TNN relies on the validity of the L\"owdin rules \cite{Lowdin1955}, that is,  for $\bm{\Phi}$ and $\bm{\widetilde{\Phi}}$ with the same dimension,
\begin{equation}
    \label{eq::Lowdin}
    \langle \mathcal{A}\left(\bm{\Phi}\right)|\mathcal{A}(\widetilde{\bm{\Phi}}) \rangle = \frac{1}{n!}|\bm{\mathcal{M}}(\bm{\Phi},\bm{\widetilde{\Phi}})|,\quad \bm{\mathcal{M}}(\bm{\Phi},\bm{\widetilde{\Phi}})=\begin{bmatrix}
\langle\phi_{1} | \widetilde{\phi}_{1} \rangle  & \dots &\langle\phi_{1} | \widetilde{\phi}_{n} \rangle \\
 \vdots & \ddots  & \vdots\\
\langle\phi_{n} | \widetilde{\phi}_{1} \rangle  & \dots &\langle\phi_{n} | \widetilde{\phi}_{n} \rangle
\end{bmatrix}, 
\end{equation}
which allows us to calculate the space inner product element-by-element in the determinant. We denote $\bm{\mathcal{M}}(\bm{\Phi}_{t_1}^{\uparrow},\bm{\Phi}_{t_2}^{\uparrow}) := \bm{\mathcal{M}}_{t_1,t_2}^{\uparrow}$ and define the biorthogonalization as $\bm{\Xi}_{t_1,t_2}^{\uparrow} = (\bm{\mathcal{M}}_{t_1,t_2}^{\uparrow})^{-1}\bm{\Phi}_{t_1}^{\uparrow}$; corresponding definitions hold for the spin-down pairs. To distinguish the components of the column vector $\bm{\Xi}_{t_1,t_2}^{\uparrow}$, we denote its $j$-th element as $\xi_{t_1,t_2,\tau_{\uparrow}(j)}(r_{\tau_{\uparrow}(j)})$. Consequently, by substituting the antisymmetric TNN ansatz Eq.~\eqref{eq::TNN_Slater} into Eq.~\eqref{eq::Rayleigh} for the Hamiltonian calculation, the computational expressions for each $(t_1,t_2)$ pair can be derived as 
\begin{equation}
\label{eq::Slater_all}
\begin{aligned}
\left\langle\Psi|\Psi\right\rangle_{t_1,t_2}&=\frac{1}{N_{\uparrow}! \cdot N_{\downarrow}!}|\bm{\mathcal{M}}_{t_1,t_2}^{\uparrow}|
|\bm{\mathcal{M}}_{t_1,t_2}^{\downarrow}|,\\
\left\langle\nabla_i\Psi|\nabla_i\Psi\right\rangle_{t_1,t_2}
&=\left\langle\Psi|\Psi\right\rangle_{t_1,t_2}\left\langle\nabla_i\xi_{t_1,t_2,i}
|\nabla_i\phi_{t_2,i}\right\rangle,\\
\left\langle\Psi\left|\frac{Q_k}{|\bm{r}_i-\bm{R}_k|}\right|\Psi\right\rangle_{t_1,t_2}&=\left\langle\Psi|\Psi\right\rangle_{t_1,t_2} \left\langle\xi_{t_1,t_2,i}\left|\frac{Q_k}{|\bm{r}_i-\bm{R}_k|}\right|\phi_{t_2,i}\right\rangle,\\
\end{aligned}
\end{equation}
and 
\begin{equation}
\label{eq::Slater_ee}
\begin{aligned}
&\left\langle\Psi\left|\frac{1}{|\bm{r}_i-\bm{r}_j|}\right|\Psi\right\rangle_{t_1,t_2}\\
=&\left\{\begin{aligned}
&\left\langle\Psi|\Psi\right\rangle_{t_1,t_2}\left\langle\begin{vmatrix}
\xi_{t_1,t_2,i}(r_i)  & \xi_{t_1,t_2,i}(r_j)\\
\xi_{t_1,t_2,j}(r_i)  & \xi_{t_1,t_2,j}(r_j)
\end{vmatrix}\left|\frac{1}{|\bm{r}_i-\bm{r}_j|}\right|\phi_{t_2,i}\phi_{t_2,j}\right\rangle,\quad \text{$i,j$ same spin}\\
&\left\langle\Psi|\Psi\right\rangle_{t_1,t_2}
\left\langle\xi_{t_1,t_2,i}\xi_{t_1,t_2,j}\left|\frac{1}{|\bm{r}_i-\bm{r}_j|}\right|\phi_{t_2,i}\phi_{t_2,j}\right\rangle,\quad \text{otherwise}.\\
\end{aligned}
\right.
\end{aligned}
\end{equation}
The detailed derivations for Eqs.~\eqref{eq::Slater_all} and \eqref{eq::Slater_ee} are provided in Appendix~\ref{app::derivation}. 
Notably, the enforced antisymmetric TNN ansatz maintains a separation-of-variables structure, inheriting the significant advantages of TNNs in handling high-dimensional integrals ~\eqref{eq::Mass}. The subsequent computation of the elementwise inner products is entirely consistent with Section~\ref{sec::Range_split}, utilizing the SOG decomposition to efficiently evaluate the Coulomb interaction contributions.

In the case of this antisymmetric TNN ansatz, the computational complexity scales as $\mathcal{O}(p^2N^3)$. Although this represents an additional $\mathcal{O}(N)$ factor compared to the soft-constrained SOG-TNN, it notably avoids the factorial complexity $\mathcal{O}(N!)$ associated with the explicit evaluation of the antisymmetrization operator.
For systems with a small number of electrons, the soft-constrained SOG-TNN scheme outlined in Algorithm~\ref{alg:algorithm_SOGTNN} incurs a lower computational overhead. However, as the system size increases, we need to face the challenge that the approximation capability of the fully separated tensor product regarding the antisymmetric space tends to diminish-a phenomenon observed in \cite{wang2025complexity}. Consequently, for systems involving more complex electron interactions, the hard-constrained SOG-TNN proposed in this subsection enables a more effective solution to the Schr\"odinger equation. In Section~\ref{sec::num_example}, we demonstrate the accuracy and efficiency of the method on the triplet states of helium atom system, and more results obtained with this version of SOG-TNN will be reported in our future work.

\section{Rate of the error convergence}\label{sec::error}
In this section, we present the way to select the concerned parameters within 
the SOG-TNN framework. For this aim, we come to analyze the energy error 
convergence with respect to the various system parameters. 

It should be noted that this analysis excludes the optimization error introduced during the neural network training, as this component is not the dominant source of the total error in most of numerical experiments \cite{wang2024solving,wang2024multi}.
Assume the real solution $\Psi_0$ and TNN-approximation $\Psi$ at each step are smooth, bounded and normalized in $L^2(\Omega)$ norm. 
We denote 
\begin{equation}\label{eq::error}
\mathcal{E}[\Psi]:=E_0-\mathscr{E}[\Psi](\bm{\theta}^{*})
\end{equation}
as the error functional with respect to the optimal TNN approximation $\Psi(\bm{r};\bm{\theta}^{*})$. This error has the following sources: 
spatial approximation error $\mathcal{E}_{\text{space}}$ arising from approximating the full solution space with the tensor space $V_{\mathcal{T}}^{p}$, SOG approximation error $\mathcal{E}_{\text{SOG}}$ due to the use of SOG approximation, and computational error $\mathcal{E}_{\text{elec}}$ due to the numerical method for calculating electron-related interaction in Section \ref{sec::Range_split}. 
It is important to note that, due to the antisymmetry by the Pauli exclusion principle, the wavefunction vanishes near the singularity when any two electrons coincide.  
Based on the strong approximation ability of neural networks
and the theoretical estimate Eq.~\eqref{eq::TNN_approx},  $\mathcal{E}_{\text{space}}$ can always be made sufficiently small provided that the TNN rank $p$ does not exceed $100$, as numerically reported in \cite{Wang2024Tensor,wang2024solving,WANG2024112928}.
Analogously to Eq.~\eqref{Hamiltonian_Operator}, under the SOG approximation the Hamiltonian operator reads 
\begin{equation} \label{eq::SOG_Hamiltonian}
\begin{aligned}
\hat{H}_{\text{SOG}} =& -\frac{1}{2}\sum_{i=1}^N \Delta_i+\sum_{k=1}^M\sum_{l=k+1}^M\frac{Q_kQ_l}{|\bm{R}_k-\bm{R}_l|}\\
&+\left[ \sum_{i=1}^N \sum_{j=i+1}^N \sum_{\ell=-\infty}^{M_1}w_{\ell}G_{\ell}(\bm{r}_i-\bm{r}_j)
-\sum_{k=1}^M\sum_{i=1}^N Q_k\sum_{\ell=-\infty}^{M_1}w_{\ell}G_{\ell}(\bm{r}_i-\bm{R}_k)\right],
\end{aligned}
\end{equation}
with the error estimate 
\begin{equation}\label{eq::error_SOG}
\begin{aligned}
\mathcal{E}_{\text{SOG}}&\le \int_{\Omega}\left | \hat{H}-\hat{H}_{\text{SOG}} \right | \left | \Psi \right |^2  \mathrm{d}\bm{r}  \\
 &\le \left[2\sqrt{2}\exp\left(-\frac{\pi^2}{2\ln b}\right)+\mathcal{\mathcal{O}}(b^{-M_1})\right](E_{ie}+E_{ee}),
\end{aligned}
\end{equation}
where the last inequality holds by Eq.~\eqref{eq::sog_error}, with the spectral convergence rate of truncation term $M_1$ similar with \cite{LIANG2025101759,Gao2024Fast}. For simplicity, here and hereafter, we neglect the dependence of the estimate on the TNN functions.

We now examine the contribution of the error term $\mathcal{E}_{\text{elec}}$. Based on the range-splitting method, $\mathcal{E}_{\text{elec}}$ can be further decomposed into three components, 
\begin{equation}\label{eq::ee_error}
\mathcal{E}_{\text{elec}}:=\mathcal{E}_{\text{short}}+\mathcal{E}_{\text{long}}+\mathcal{E}_{\text{mid}},
\end{equation}
which are due to
the short-range asymptotic, the long-range Chebyshev truncation, and the mid-range model reduction, respectively. And these errors can be analyzed independently.

We first examine the error $\mathcal{E}_{\text{short}}$, which arises due to the use of the leading asymptotics in Eq. \eqref{eq::short_approx_integral}. Lemma~\ref{lemma::asymptotic} quantifies the asymptotic error of the approximation.
\begin{lemma}
    \label{lemma::asymptotic}
    For one-dimensional analytic and bounded function $u(x)$, and a centralized normal distribution $g(x)=1/(\sqrt{2\pi}s)\exp(-x^2/(2s^2))$,
    then 
    \begin{equation}
        \label{eq::lemma1}
        \left|\int_{\mathbb{R}} u(y)\cdot g(x-y)\mathrm{d}y-u(x)\right|\le \sum_{k =1}^{+\infty } \frac{\left | u^{(2k )}(x) \right | }{(2k )!!} s^{2k }
    \end{equation}
    holds for all $x\in \mathbb{R}$.
\end{lemma}
\begin{proof}
The asymptotic error can be estimated by
\begin{equation}\label{eq::short_estimate}
\begin{aligned}
&\left|\int_{\mathbb{R}} u(y)\cdot g(x-y)\mathrm{d}y-u(x)\right|=\left|\int_{\mathbb{R}} (u(y)-u(x))\cdot g(x-y)\mathrm{d}y\right|\\&=\sum_{k =1}^{+\infty} \frac{ u^{(2k )}(x) }{(2k )!} \cdot\int_{\mathbb{R}}\frac{(y-x)^{2k }}{\sqrt{2\pi}s}e^{-(y-x)^2/2s^2}\mathrm{d}y\\&\leq \sum_{k =1}^{+\infty } \frac{ |u^{(2k )}(x)|  }{(2k )!!} s^{2k }.\\
\end{aligned}
\end{equation}
The last inequality obviously holds after exactly calculating all the moments. We then complete the proof. 
\end{proof}
Applying Lemma~\ref{lemma::asymptotic} into Eq.~\eqref{eq::Short_contribution}, we can derive an absolute error estimate
\begin{equation}\label{short_error_total}
\begin{aligned}
\mathcal{E}_{\text{short}}&\le\sum_{\ell\in\mathcal{M}_{\text{short}}}w_{\ell}
\sum_{t_1,t_2=1}^{p}|\alpha_{t_1}\alpha_{t_2}|\sum_{i<j}^{N}\left[|\mathcal{J}_{t_1,t_2,i,j}|
\cdot\sum_{k =1}^{+\infty}\frac{s_{\ell}^{2k }}{(2k )!!} |\mathcal{I}_{t_1,t_2,i,j}^{\eta,(2k )}|\right]
=\mathcal{O}(b^{4M_{\text{sm}}}),
\end{aligned}
\end{equation}
where $M_{\text{sm}}$ is the critical index that belongs to $\mathcal{M}_{\text{short}}$ and $M_{\text{sm}}+1$ belongs to $\mathcal{M}_{\text{mid}}$, and
\begin{equation}\label{eq::modulo}
|\mathcal{I}_{i,j}^{\eta,(2k )}|:=\left| \int_{\Omega_0}\phi_{i,t_1}^{\eta}(\eta_i)
\phi_{i,t_2}^{\eta}(\eta_i)\phi_{j,t_1}^{\eta,(2k )}(\eta_i)\phi_{j,t_2}^{\eta,(2k )}(\eta_i)d\eta_i \right|.
\end{equation}
This result demonstrates that the error in the total energy exhibits an exponential decay 
as the short-range partitioning index $M_{sm}$ decreases. This same decay rate holds 
for the estimate of the short-range ion-electron interaction term, 
which is computed according to Eqs.~\eqref{eq::E_ie} and \eqref{eq::ie_split}.

Regarding the truncation error of the long-range Chebyshev expansion $\mathcal{E}_{\text{long}}$,
By substituting the error estimate from Eq.~\eqref{eq:cheb_error} into the expression for the
long-range contribution in Eq.~\eqref{eq::Long_contribution}, we can deduce that
\begin{equation}
\label{long_error_total}
\begin{aligned}
\mathcal{E}_{\text{long}}&\le\sum_{\ell\in\mathcal{M}_{\text{long}}}w_{\ell}\sum_{i<j}^{N}
\int_{\Omega}\left|G_{\ell}\left(\eta_i-\eta_j\right)- \sum_{|m+n|\le S} C_{m,n}^{r_c}(\ell) 
T_m\left(\frac{\eta_i}{r_c}\right)T_n\left(\frac{\eta_j}{r_c}\right)\right|\left | \Psi \right |^{2} d\bm{r}\\
&\le \sum_{\ell\in\mathcal{M}_{\text{long}}}w_{\ell}\frac{N(N-1)r_c^S}{2\sqrt{2}(2s_{\ell} S^{1/2})^{S}}\\
&= \mathcal{O}\left(\frac{N^2r_c^{S}}{(2\sigma S^{1/2})^{S}}b^{(2-S)M_{ml}}\right),
\end{aligned}
\end{equation}
where $M_{ml}$ is the critical index that belongs to $\mathcal{M}_{\text{long}}$ and $M_{ml}-1$ belongs to $\mathcal{M}_{\text{mid}}$. 
Based on the a priori long-range error estimate in Eq.~\eqref{long_error_total}, one can adaptively choose the order of the Chebyshev expansion $S$ for each $\ell\in \mathcal{M}_{\text{long}}$.

Finally, for the mid-range component, the error $\mathcal{E}_{\text{mid}}$ originates from the truncation of the singular value matrix $\bm{\Sigma}_\ell$.   Assume that a $K\times K$ symmetric matrix $\bm{A}$ has the SVD decomposition $\bm{A}=\bm{U}\bm{\Sigma}\bm{V}^\top$, where $\bm{\Sigma}=\emph{\rm diag}\{\lambda_1,\cdots,\lambda_K\}$ is arranged in the descending order, and $\bm{U}$ and $\bm{V}$ are orthogonal matrices. Suppose that it is truncated as $\bm{\Sigma}_{r}=\emph{\rm diag}\{\lambda_1,\cdots,\lambda_r,0,\cdots,0\}$ with $r<K$, then for any $\bm{x},\bm{y}\in \mathbb{R}^{K}$, one has the error estimate \cite{golub2013matrix}
\begin{equation}\label{eq::estimate_SVD}
\left|\bm{x}^\top\bm{A}\bm{y}-\bm{x}^\top(\bm{U}\bm{\Sigma}_r\bm{V}^\top)\bm{y}\right|\le \lambda_{r+1}\|\bm{x}\|_{2}\|\bm{y}\|_{2},
\end{equation}
where $\|\cdot\|_2$ denotes the Euclidean 2-norm of $\mathbb{R}^{K}$. By applying it to the practical computation of the mid-range contributions in Eqs.~\eqref{eq::mid_2Dintegral} and \eqref{eq::Mid_contribution}, the relative accuracy requirement is satisfied simply by ensuring that the first neglected singular value $\lambda_{r_{\ell}+1}^{\ell}$ is less than the prescribed tolerance $\epsilon$. Since this reduction is independent of the network during each optimization step, it can be performed as pre-computation for all mid-range Gaussian matrices $\bm{G}_\ell$ with $\ell\in \mathcal{M}_{\text{mid}}$.

With these results, we now have all the necessary components for a complete parameter selection strategy. By combining the various error decay rates--namely, that of the TNN spatial approximation Eq.~\eqref{eq::TNN_approx} and those for the short- and long-range SOG-TNN components in Eqs.~\eqref{eq::error_SOG}, \eqref{short_error_total}, and \eqref{long_error_total}--with the objective of minimizing computational complexity of Eq.~\eqref{eq::total_complexity}, we can numerically determine all system parameters. 
These include the TNN rank $p$; the SOG parameters $b, \sigma$, and $M_1$; 
the range-splitting thresholds $M_{sm}$ and $M_{ml}$; and the Chebyshev expansion order $S$, 
all chosen to satisfy a prescribed accuracy requirement. 
In Section~\ref{sec::num_example}, we will apply these findings to solve the many-body Schr\"{o}dinger equation for several atomic systems. The numerical results will demonstrate that the SOG-TNN framework achieves high-precision solutions for the ground-state energy and wavefunction, and does so with a significant reduction in both computational complexity and memory overhead. This outcome highlights the potential of our method for tackling larger and more complex high-dimensional quantum problems.

\section{Numerical results}\label{sec::num_example}
In this section, we present numerical results for the SOG-TNN applied 
to several atomic systems: 
hybrogen (H), helium (He and $(^{3}\text{S}) \text{He}$), 
lithium (Li), and beryllium (Be).
These results are used to demonstrate the advantages of our new algorithm with respect to accuracy, computational efficiency, and memory overhead. In the calculations, we employ a TNN architecture where each internal feedforward sub-neural network is constructed with two hidden layers. The sine function is used as the activation function $\eth(x)=\sin(x)$ throughout, and the Adam optimizer \cite{kingma2014adam} is applied for the stochastic gradient descent method. All the calculations were performed on a single card with a NVIDIA A800 GPU  (80GB memory, 5120 bit interface, 2039GB/s bandwidth) or a RTX 4090D GPU (24GB memory, 384 bit interface, 1040GB/s bandwidth), both provided by the Academy of Mathematics and Systems Science, Chinese Academy of Sciences.

\subsection{Accuracy results for atom systems}\label{sec :: atom system}
We first validate the reliability and high accuracy of the SOG-TNN algorithm 
by calculating the hydrogen atom, which has a known analytical 
solution \cite{cohen1973quantum}. 
Figure~\ref{fig::H}(a-b) plots the training error for the energy 
and the radial distribution of the wavefunction in an initial 
learning rate $\gamma=10^{-4}$.
With a constant learning rate, there is a clear rebound in the error curve after $1.2\times 10^6$ epochs of training. This nonphysical phenomenon was also encountered in existing machine learning methods for solving the high-dimensional Schr\"{o}dinger equation \cite{202209v2}. To address this issue, we employ an adaptive, piecewise learning rate decay schedule. Whenever a rebound in the error is detected, we restart the training from an early stage with a smaller learning rate to perform a finer optimization. The results using this strategy are shown in Figure~\ref{fig::H}(c-d). Evidently, the errors in both the system energy and the radial distribution function decay more stably and reach a higher level of accuracy over $10^{-7}$ and $10^{-4}$, respectively. These results demonstrate the effectiveness of the SOG-TNN algorithm.

Next, we present numerical tests for the He, Li and Be systems. Based on rough radial distribution of the wavefunction obtained at the end of the first training process with the initial learning rate, we partitioned the one-dimensional domain \(\Omega_0\) into subintervals corresponding to a certain ratio, concentrating more quadrature points near the density peaks. On each subinterval we apply a composite Gauss-Legendre rule with \(N_{\rm panel}\) equally-spaced panels and \(N_{\rm quad}\) nodes per panel. For all applied SOG in this work, $\sigma=1$ is fixed, and the maximum Chebyshev expansion order is also selected to be $S=50$.  Table \ref{table::SOG parameter} lists the SOG parameters at cutoff $r_c=1~a.u.$, including relative error tolerance $\epsilon$, SOG base $b$, critical indices of range-splitting $\left(M_{sm}^{ie},M_{sm}^{ee}\right )$ for $\left( E_{ie}^{\text{short}},E_{ee}^{\text{short}}\right)$, $M_{ml}$ for $E_{ee}^{\text{long}}$, and $M_1$ for right truncation number of the BSA.  The TNN parameters are present in Table \ref{table::TNN parameter}, where the corresponding SOG parameters can be obtained using Table~\ref{table::SOG parameter} for different $r_c$.
\begin{figure}[!ht] 
\centering   
\includegraphics[width=1.0\textwidth]{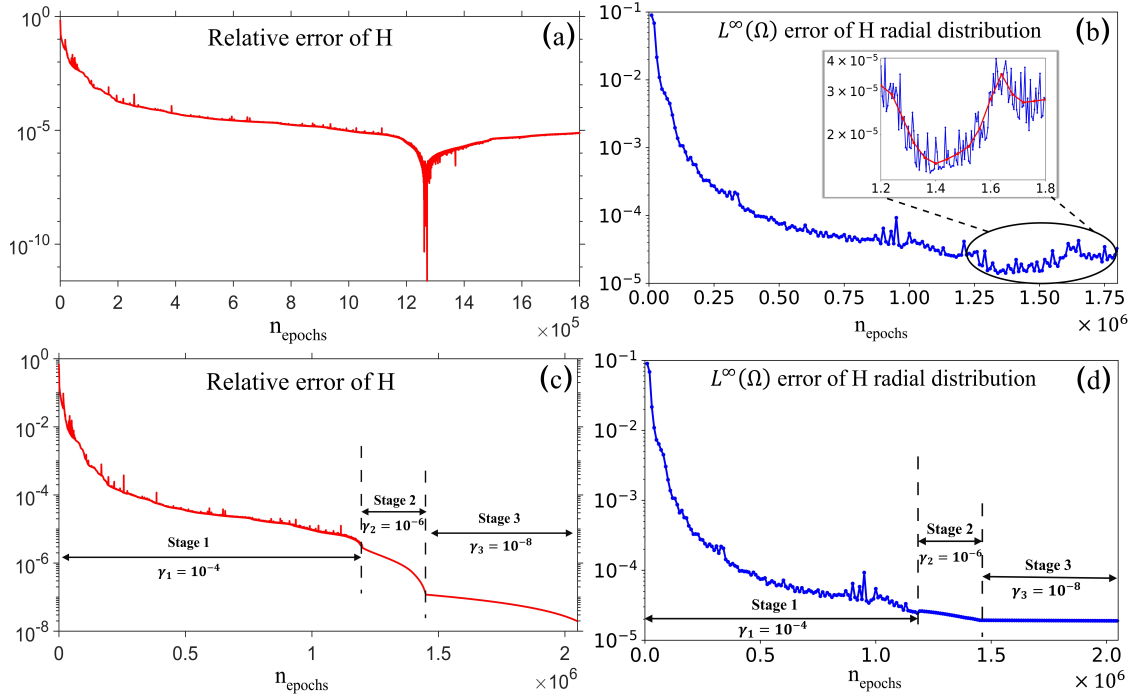}
\caption{Computation results for the Hydrogen atom by the SOG-TNN algorithm, comparing performance without and with the piecewise learning rate decay schedule. (a,b) results without the decay schedule, and (c,d) results with the decay schedule. The training curves in (a,c) shows the relative error of the ground state energy $E_0$, and (b,d) shows the $L^{\infty}(\Omega)$ error of radial distribution of the ground state wavefunction $\Psi_0$.}
\label{fig::H}
\end{figure}

\begin{table}[h!]
\centering
\caption{Parameter sets for the SOG approximation with different accuracy at cutoff $r_c=1~a.u.$. 
}
\label{table::SOG parameter}
\begin{tabular}{c|c|c|c|c|c}
\toprule
\ \ \ $\epsilon$\ \ \  &\ \ \  $b$ \ \ \ &\ \ \  $M_{sm}^{ie}$ \ \ \ &\  \ \ $M_{sm}^{ee}$ \ \ \ & \ \ \ $M_{ml}$\ \ \  &\ \ \ $M_1$  \ \ \ \\
\midrule
 \ \ \ $10^{-11}\ \ \ $ & \ \ \  $1.20$  \ \ \ & \ \ \  $-44$  \ \ \ & \ \ \  $-34$  \ \ \ & \ \ \  $-6$  \ \ \ &  \ \ \ $156$ \ \ \  \\
 \ \ \ $10^{-7}\ \ \ $ & \ \ \  $1.30$  \ \ \ & \ \ \  $-25$  \ \ \ & \ \ \  $-18$  \ \ \ & \ \ \  $-5$  \ \ \ &  \ \ \ $80$ \ \ \  \\
 $10^{-6}$ & $1.35$ & $-20$ & $-14$ & $-5$ & $73 $ \\
 $10^{-5}$ & $1.40$ & $-18$ & $-12$ & \ \ \  $-5$  \ \ \ & $60$ \\
\bottomrule
\end{tabular}
\end{table}

\begin{table}[h!]
\centering
\caption{Parameters selected in TNN for different atoms. The symbol '$\backslash$' means no corresponding parameter in the system.}
\label{table::TNN parameter}
\begin{tabular}{c|c|c|c|c|c|c|c|c}
\toprule
System & $\epsilon$ & $r_c~(a.u.)$ & TNN sizes & $( \lambda_{ij}^{\uparrow},\lambda_{ij}^{\downarrow})$ & $K$ & ratio & $N_{\text{panel}}$ & $N_{\text{quad}}$ \\
\midrule
He & $10^{-7}$ &  $10$  & $\left[ 1,50,50,60 \right]$ & $(\backslash\ ,\ \backslash)$ & $340$ & $\left[ 2,1,2\right]$ & $\left[ 10,30,10\right]$ & $\left[ 8,6,8\right]$  \\
Li & $10^{-7}$ & $12$  & $\left[ 1,50,50,60 \right]$ & $(50,\ \backslash)$ & $340$ & $\left[ 2,1,2\right]$ & $\left[ 10,30,10\right]$ & $\left[ 8,6,8\right]$   \\
Be & $10^{-5}$ & $15$  & $\left[ 1,50,50,60 \right]$ & $(50,50)$ & $310$ & $\left[ 1,2,1\right]$ & $\left[ 8,25,8\right]$ & $\left[ 10,6,10\right]$ \\
\bottomrule
\end{tabular}
\end{table}

Consider the helium atom. Since it has two electrons with opposite spins, the spatial part of its wavefunction is symmetric, and thus no antisymmetry constraint is imposed on the SOG-TNN. Figure~\ref{fig::He} shows the relative error training curve of the energy as well as the radial distribution of the wavefunction. The results demonstrate that the SOG-TNN method stably computes a high-precision ground-state energy for the helium atom and accurately captures the radial distribution of its wavefunction. The calculated energy by the SOG-TNN is $E_0=-2.9037241864374510~a.u.$, achieving $10^{-7}$ accuracy in relative error in comparison with the reference solution $E_0^{\text{He}}=-2.90372$
$43770341144~a.u.$ \cite{thakkar1994ground}.

\begin{figure}[!ht] 
    \centering
    \includegraphics[width=1.0\textwidth]{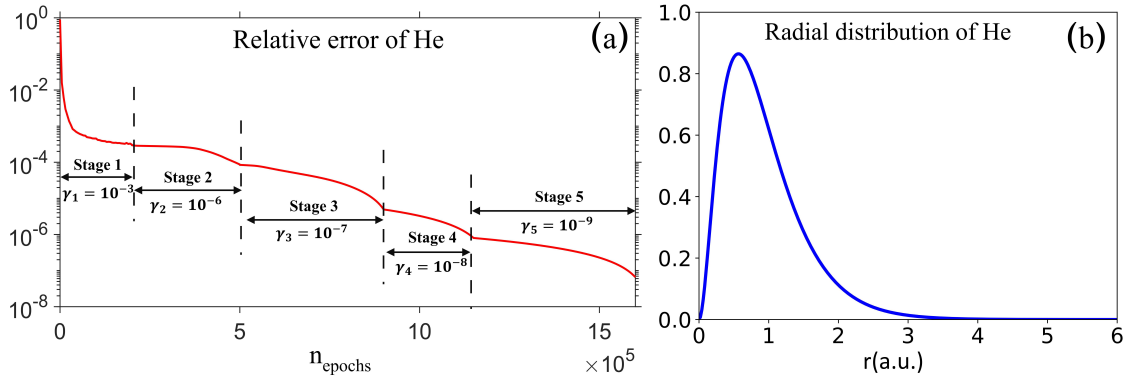}
    \caption{(a) The relative error training curve of the ground state energy and (b) The radial distribution of the ground state wavefunction of the helium atom by the SOG-TNN algorithm. }
    \label{fig::He}
\end{figure}

Lithium and beryllium belong to the second period and present a greater challenge than helium due to their larger number of degrees of freedom and more complex properties. Since they are many-electron systems, Li and Be must satisfy the Pauli exclusion principle. Consequently, the SOG-TNN framework requires the inclusion of an antisymmetry penalty term in the loss function. For these calculations, the penalty factors and corresponding TNN parameters are listed in the second and third rows of Table~\ref{table::TNN parameter}, respectively.
The accuracy targets for the Li and Be atoms are set differently based on the quality of available benchmark data. For Li, a high-precision reference exists, thus we require an error tolerance of~$10^{-7}$. For Be, however, the precision is limited by the reference solution, so our target is relaxed to $10^{-5}$. The corresponding computational results are presented in Figures~\ref{fig::Li} and~\ref{fig::Be}, respectively. The SOG-TNN algorithm stably computes high-precision ground-state energies for both the Li and Be systems. As shown in Figures~\ref{fig::Li}(b) and~\ref{fig::Be}(b), the calculated wavefunction radial distributions exhibit the multi-peak structure of these well-known multi-electron systems. 
The calculated energies are $E_0=-7.4780595868441369~a.u.$ for the Li atom and $E_0=-14.6672437863217056~a.u.$ for the Be atom, both meeting the target 
accuracy with the referenced energy $E_0^{\text{Li}}=-7.478060326(10)~a.u. $ \cite{mckenzie1991variational} and $E_0^{\text{Be}}=-14.66736~a.u.$ \cite{davidson1991ground}, respectively. Furthermore, the error of the antisymmetric term can be formulated by$$\mathcal{E}_{\mathcal{A}}\left[ \Psi \right]=\max_{ij}\left\{\frac{\langle T_{\tau_{\uparrow}(i)\tau_{\uparrow}(j)}\Psi|\Psi\rangle}{\langle\Psi|\Psi\rangle}, \frac{\langle T_{\tau_{\downarrow}(i)\tau_{\downarrow}(j)}\Psi|\Psi\rangle}{\langle\Psi|\Psi\rangle}\right\}+1.$$
Figure~\ref{fig::Ant_check} presents the training curves of the antisymmetry 
penalty in results of Figs.~\ref{fig::Li} and \ref{fig::Be} for Li and Be atoms, 
which illustrates a high-precision convergence level with $n_\mathrm{epochs}$. 
The purple line in Panel (a) displays the evolution of the antisymmetry 
error under another fixed hyperparameter setting of $\lambda_{ij}^{\uparrow}=1$, 
which reveals the accuracy of the antisymmetry can be effectively 
regulated by tuning this hyperparameter.
These results demonstrate that the penalty constraint in SOG-TNN 
effectively prevents convergence to unphysical solutions.
\begin{figure}[!ht] 
\centering
\includegraphics[width=1.0\textwidth]{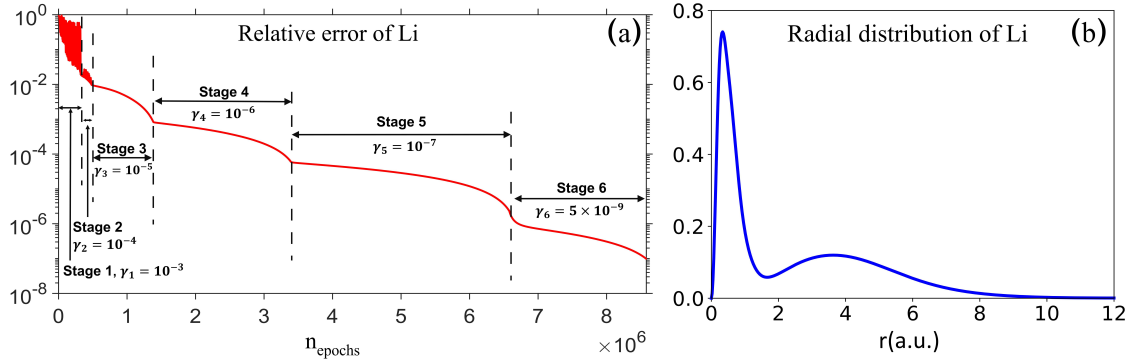}
\caption{(a) Relative error training curve of the ground state energy, and (b) radial distribution of the ground state wavefunction of the lithium atom by the SOG-TNN algorithm.}
\label{fig::Li}
\end{figure}

\begin{figure}[!ht] 
\centering   \includegraphics[width=1.0\textwidth]{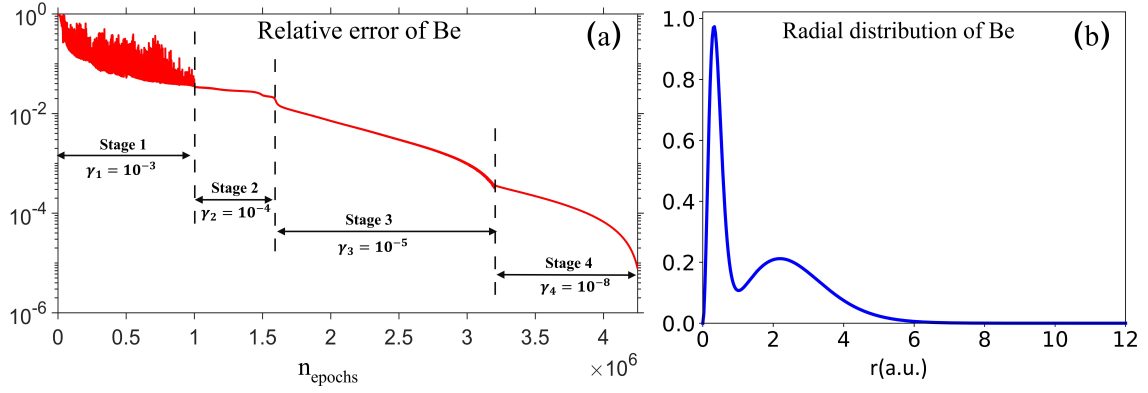}
\caption{(a) Relative error training curve of the ground 
state energy, and (b) radial distribution of the ground 
state wavefunction of the beryllium atom by the SOG-TNN algorithm.}\label{fig::Be}
\end{figure}

\begin{figure}[!ht] 
\centering
\includegraphics[width=1.0\textwidth]{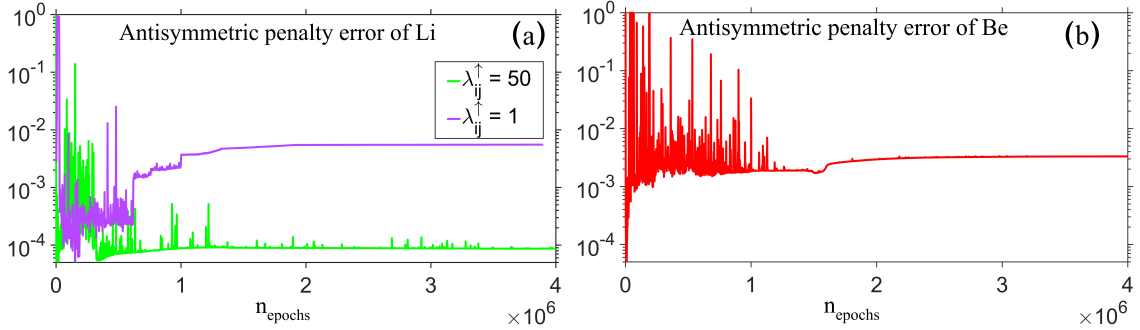}
\caption{Evolution of the antisymmetric penalty error of (a) the lithium system in Figure~\ref{fig::Li}, and (b) the beryllium system in Figure~\ref{fig::Be} by the SOG-TNN algorithm. In the lithium case (a), the penalty parameter $\lambda_{ij}^{\uparrow}$ takes $1$ and $50$, respectively, to provide different strengths of antisymmetry enforcement. }
\label{fig::Ant_check}
\end{figure}

To further demonstrate the accuracy of the SOG-TNN, we test the antisymmetric TNN ansatz on the triplet states of helium $\left( (^{3}\text{S}) \text{He} \right )$, which has two electrons of the same spin, corresponding to the lowest fully antisymmetric states. 
The same TNN parameters as the ground state listed in Table~\ref{table::TNN parameter} are adopted, with the exception of an enlarged cutoff radius $r_c=15$ to cover the $2\text{S}$ orbital, and the results are present in Figure~\ref{fig::He_triplet}.
It is clearly shown that the SOG-TNN method with enforced antisymmetry stably computes a high-precision energy and effectively captures the physical features of the triplet state. The calculated energy is $E_0=-2.1752292661984516~a.u.$, achieving $10^{-7}$ accuracy in relative error, compared to the reference value $- 2.1752293782367913~a.u.$ in Ref. \cite{Nakashima2008solving}.

\begin{figure}[!ht] 
\centering   \includegraphics[width=1.0\textwidth]{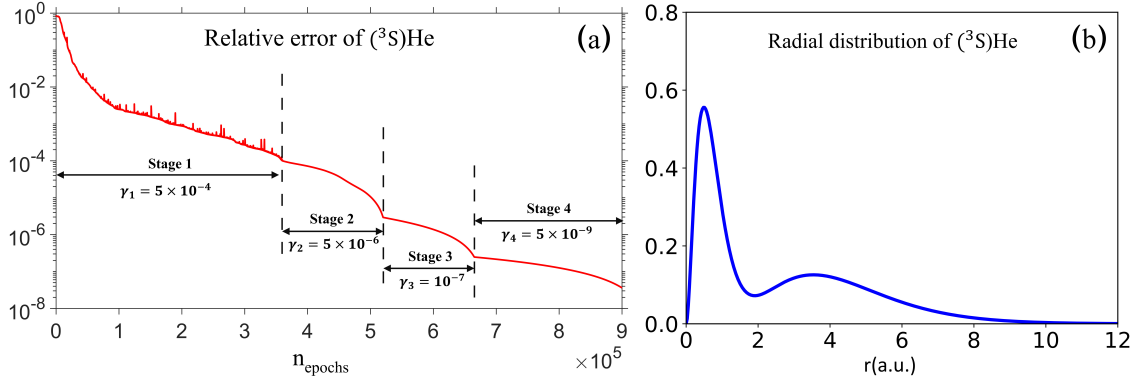}
\caption{(a) Relative error training curve of the triplet state energy, and (b)  radial distribution of the triplet state wavefunction of the helium atom by the SOG-TNN algorithm with enforced antisymmetry. }
\label{fig::He_triplet}
\end{figure}

\subsection{Efficiency of the SOG-TNN method}
We first study the performance advantage offered by the SOG-TNN method on the time and memory cost under the same level of accuracy, by comparing the results with those from the SHE-TNN architecture  \cite{202209v2} based on the spherical harmonic expansion.  All comparison experiments for the He, Li, and Be systems were performed on a single NVIDIA A800 GPU, and a summary of this comparison is presented in Table~\ref{table::Comp}. As pointed out in Section~\ref{sec::SOG-TNN}, the SOG method effectively handles the Coulomb singularity, leading to improved solution accuracy. Moreover, the inherent dimensional separability of Gaussian functions aligns perfectly with the tensor product structure, which facilitates rapid computation and reduces memory requirements. The data in Table~\ref{table::Comp} clearly demonstrate these advantages of the SOG-TNN method, which become more pronounced as system complexity increases. For the most complex system in our experiments, the Be atom, the SHE-TNN method was severely limited by memory constraints; at the limit of a single GPU, it could only achieve an accuracy worse than~$10^{-2}$. In contrast, our proposed method improves this accuracy by more than three orders of magnitude while consuming only one-tenth of the GPU memory. This result indicates that, for high-precision calculations, the SOG-TNN algorithm has the potential on a single GPU to handle systems beyond the size limits of the original approach. Furthermore, it suggests that the introduction of multi-GPU parallelization could enable the accurate computation of even larger systems that are currently challenging.

\begin{table}[htbp]
\centering
\caption{Comparison of the computational accuracy, memory consumption, and time per 
step between the SOG-TNN and the SHE-TNN. The parameters for each test system selected 
by the SOG-TNN architecture are consistent with Tables~\ref{table::SOG parameter} 
and \ref{table::TNN parameter}.}
\label{table::Comp}
\begin{tabular}{c|cc|cc|cc}
\toprule
\multirow{2}{*}{System} & \multicolumn{2}{c|}{Relative error $\mathcal{E}$} & \multicolumn{2}{c|}{Memory (MiB)} & \multicolumn{2}{c}{Time per step ($ms$)} \\
\cmidrule(lr){2-3} \cmidrule(lr){4-5} \cmidrule(lr){6-7} 
                         &  SHE-TNN & SOG-TNN & SHE-TNN & SOG-TNN & SHE-TNN & SOG-TNN \\
\midrule
He        & $8.38\times 10^{-8}$   & $6.56\times 10^{-8}$ & $26338$ & $3790$ & $132$   & $69$    \\
\midrule
Li       & $7.78\times 10^{-7}$   & $9.80\times 10^{-8}$ & $49332$ & $9070$ & $285$   & $142$     \\
\midrule
Be       & $3.98\times 10^{-2}$   & $7.68\times 10^{-6}$ & $75646$ & $8528$ & $476$   & $199$     \\
\bottomrule
\end{tabular}
\end{table}

Finally, as a nonlinear, adaptive sparse grid method, a comparison between the SOG-TNN and traditional sparse grids~\cite{griebel2009tensor} is also meaningful. We consider the ground-state benchmark in Section~\ref{sec :: atom system}.
To ensure a rigorous comparison, we interpret the performance differences through both approximation capacity and constraint satisfaction in Table~\ref{table::Comp SG}. The helium test case serves as a control experiment, as it does not involve antisymmetry constraints. Here, the SOG-TNN achieves remarkable accuracy with far less degrees of freedom, establishing a baseline for its spatial approximation capability. Extending the method to the lithium and beryllium systems, we observe that the SOG-TNN maintains this advantage even under antisymmetry constraints. Although the soft constraint introduces a small additional error, SOG-TNN outperforms the SG method in energy accuracy within the antisymmetry error bound.
Finally, to guarantee a fair comparison with identical constraints, we turn to the helium triplet state where both methods strictly enforce antisymmetry. The SOG-TNN improves the accuracy by near two orders of magnitude over the SG, while reducing the basis size by over two orders of magnitude.
These results indicate that either the soft or the hard constraint achieves high-precision solutions without violating physical constraints.

\begin{table}[htbp]
% \color{blue}
% \arrayrulecolor{blue} 
\centering
\caption{Comparison of results between the SG and the SOG-TNN. The SG results are obtained from \cite{griebel2009tensor}. The symbol '$\backslash$' means no antisymmetry constraints for He atom  and the symbol of $\checkmark$ denotes the strict satisfaction of the antisymmetry. Basis size represents the number of expansion terms in the variational ansatz which corresponds to the rank parameter $p$ within the SOG-TNN framework. Parameters for the SOG-TNN are consistent with Tables~\ref{table::SOG parameter} and \ref{table::TNN parameter}.}
\label{table::Comp SG}
\begin{tabular}{@{}c|ccc|cc|cc@{}}
\toprule
\multirow{2}{*}{system} & \multicolumn{3}{c|}{Ground state energy} & \multicolumn{2}{c|}{Antisymmetric error} & \multicolumn{2}{c}{Basis size}                      \\ \cmidrule(l){2-8} 
                        & SG        & SOG-TNN      & Benchmark      & SG               & SOG-TNN                & SG                           & SOG-TNN            \\ \midrule
He                      & -2.903298 & -2.90372418  & -2.90372438    & $\backslash$     & $\backslash$           & \multicolumn{1}{c|}{32,608} & \multirow{5}{*}{60} \\ \cmidrule(r){1-7}
Li                      & -7.477022 & -7.47805958  & -7.47806032    & $\checkmark$     & $8.68\times 10^{-5}$   & \multicolumn{1}{c|}{22,203} &                      \\ \cmidrule(r){1-7}
Be                      & -14.659783& -14.66724378 & -14.66736      & $\checkmark$     & $3.32\times 10^{-3}$   & \multicolumn{1}{c|}{24,775} &                      \\ \cmidrule(r){1-7}
($^3\text{S}$)He               & -2.175220 & -2.17522927  & -2.17522938    & $\checkmark$     & $\checkmark$           & \multicolumn{1}{c|}{20,386} &                      \\ \bottomrule
\end{tabular}
\end{table}

\section{Conclusion}
\label{sec::conclusion}
In this work, we have proposed a novel SOG-TNN algorithm the direct solution of the high-dimensional Schr\"{o}dinger equation. 
The core contribution of our SOG-TNN approach is the introduction of the SOG decomposition for the Coulomb potential, leading
to a compatible representation with the low-rank tensor product structure of the TNN, and thus a dimension-by-dimension
separation of the ion-electron and electron-electron interactions. 
We further developed a range-splitting scheme based on Gaussian bandwidths, enabling us to treat 
the resultant short-, long-, and mid-range components with specialized, efficient low-rank techniques, 
including asymptotic analysis, Chebyshev expansion, and model reduction. 
The SOG-TNN robustly handles the computational challenges from 
the non-separability and singularity of the Coulomb kernel, significantly reducing the overall computational complexity. Our theoretical 
analysis of the error convergence provides a systematic strategy for tuning all parameters such that a prescribed 
precision is achieved. Numerical results for atomic systems validate 
the high accuracy and efficiency of the SOG-TNN, and demonstrate attractive performance compared to those of existing methods.

The SOG-TNN algorithm is shown to be promising for broader applications. 
Within the Rayleigh principle, it is naturally suited for the computation of excited states. 
The inherent dimensional separability of the method makes it readily applicable to 
anisotropic systems, while model reduction of the Gaussian approximation offers a clear 
path for further optimization. Future work will also include the implementation of a fully optimized SOG-TNN package such that the ground- and excited-state properties of various
complex atomic and molecular structures with high precision and efficiency can be solved.

\section*{Acknowledgement}

The work of Q. Zhou and Z. Xu are supported by National Natural Science Foundation of China (grants No. 12426304, 12325113 and 125B2023). The work of T. Wu, J. Liu and Q. Sun are supported by Zhiyuan Future Scholar Program (grant No. ZIRC2025-01). The work of H. Xie is supported by the Strategic Priority Research Program of the Chinese 
Academy of Sciences (XDB0620203, XDB0640000, XDB0640300), 
 National Key Laboratory of Computational Physics (grant No. 6142A05230501), 
 National Key Research and Development Program of China (2023YFB3309104), 
National Natural Science Foundations of China (grant No. 1233000214), 
Science Challenge Project (TZ2024009), National Center for Mathematics and Interdisciplinary Science, CAS. The work is partially supported by the SJTU Kunpeng \& Ascend Center of Excellence.

\appendix
\section{Derivations of Eqs.~\eqref{eq::Slater_all} and \eqref{eq::Slater_ee}}
\label{app::derivation}
We start from biorthogonalization $\bm{\Xi}_{t_1,t_2}^{\uparrow} = (\bm{\mathcal{M}}_{t_1,t_2}^{\uparrow})^{-1}\bm{\bm{\Phi}}_{t_1}^{\uparrow}$, which provides the following two properties. First, the antisymmetrizations of $\bm{\Xi}_{t_1,t_2}^{\uparrow}$ and $\bm{\bm{\Phi}}_{t_1}^{\uparrow}$ are the same up to a constant $\left|\bm{\mathcal{M}}_{t_1,t_2}^{\uparrow}\right|^{-1}$, that is, 
\begin{equation}
\begin{aligned}
\mathcal{A}\left(\bm{\Xi}_{t_1,t_2}^{\uparrow}\right) &= \frac{1}{N_{\uparrow}!} \det\left[\left((\bm{\mathcal{M}}_{t_1,t_2}^{\uparrow})^{-1}\bm{\bm{\Phi}}_{t_1}^{\uparrow}\right)(\bm{r}_{\tau_{\uparrow}(1)}) , \dots, \left((\bm{\mathcal{M}}_{t_1,t_2}^{\uparrow})^{-1}\bm{\bm{\Phi}}_{t_1}^{\uparrow}\right)(\bm{r}_{\tau_{\uparrow}(N_{\uparrow})}) \right] \\
&= |(\bm{\mathcal{M}}_{t_1,t_2}^{\uparrow})|^{-1} \frac{1}{N_{\uparrow}!} \det \left[\bm{\bm{\Phi}}_{t_1}^{\uparrow}(\bm{r}_{\tau_{\uparrow}(1)}) ,\dots,\bm{\bm{\Phi}}_{t_1}^{\uparrow}(\bm{r}_{\tau_{\uparrow}(N_{\uparrow})}) \right]\\
&= |(\bm{\mathcal{M}}_{t_1,t_2}^{\uparrow})|^{-1}  \mathcal{A}\left(\bm{\bm{\Phi}}_{t_1}^{\uparrow}\right),
\end{aligned}
\end{equation}
and corresponding definitions hold for the spin-down pairs. Here $\det(\bm{v}_1,\cdots,\bm{v}_n)$ denotes the determinant of the $n$-dim colomn vectors $\bm{v}_1,\cdots,\bm{v}_n$.
The second property is about the biorthogonality of $\bm{\Xi}_{t_1,t_2}^{\uparrow}$ with respect to $\bm{\bm{\Phi}}_{t_2}^{\uparrow}$, that is, 
\begin{equation}
\begin{aligned}
\bm{\mathcal{M}}(\bm{\Xi}_{t_1,t_2}^{\uparrow},\bm{\bm{\Phi}}_{t_2}^{\uparrow})&=\begin{bmatrix}
\langle\xi_{t_1,t_2, \tau_{\uparrow}(1)} |{\phi}_{t_2,\tau_{\uparrow}(1)} \rangle  & \dots &\langle\xi_{t_1,t_2, \tau_{\uparrow}(1)} | {\phi}_{t_2,\tau_{\uparrow}(N_{\uparrow})} \rangle \\
\vdots & \ddots  & \vdots\\
\langle\xi_{t_1,t_2, \tau_{\uparrow}(N_{\uparrow})} | {\phi}_{t_2,\tau_{\uparrow}(1)} \rangle  & \dots &\langle\xi_{t_1,t_2, \tau_{\uparrow}(N_{\uparrow})} | {\phi}_{t_2,\tau_{\uparrow}(N_{\uparrow})} \rangle
\end{bmatrix}\\
& = (\bm{\mathcal{M}}_{t_1,t_2}^{\uparrow})^{-1} \begin{bmatrix}
\langle\phi_{t_1, \tau_{\uparrow}(1)} |{\phi}_{t_2,\tau_{\uparrow}(1)} \rangle  & \dots &\langle\phi_{t_1, \tau_{\uparrow}(1)} | {\phi}_{t_2,\tau_{\uparrow}(N_{\uparrow})} \rangle \\
\vdots & \ddots  & \vdots\\
\langle\phi_{t_1, \tau_{\uparrow}(N_{\uparrow})} | {\phi}_{t_2,\tau_{\uparrow}(1)} \rangle  & \dots &\langle\phi_{t_1,\tau_{\uparrow}(N_{\uparrow})} | {\phi}_{t_2,\tau_{\uparrow}(N_{\uparrow})} \rangle
\end{bmatrix}\\
& = (\bm{\mathcal{M}}_{t_1,t_2}^{\uparrow})^{-1}\bm{\mathcal{M}}(\bm{\bm{\Phi}}_{t_1}^{\uparrow},\bm{\bm{\Phi}}_{t_2}^{\uparrow})  \\
&= \bm{I},
\end{aligned}
\end{equation}
and corresponding definitions hold for the spin-down pairs.
Next, combining with L\"owdin rules and biorthogonality, we derive the Eqs.~\eqref{eq::Slater_all} and \eqref{eq::Slater_ee} element by element.

The kinetic part $\left(E_{\text{K}}\right)_{t_1,t_2}$ can be calculated by
\begin{equation}
\begin{aligned}
\frac{1}{2}\sum_{i}^N \left\langle \nabla_{i}\Psi | \nabla_{i} \Psi \right \rangle_{({t_1,t_2})} &= \frac{1}{2N_{\downarrow}!}\left(\sum_{{i=1}}^{N_{\uparrow}} \left \langle \nabla_{\tau_{\uparrow}(i)}\mathcal{A}(\bm{\bm{\Phi}}_{t_1}^{\uparrow})|\nabla_{\tau_{\uparrow}(i)} \mathcal{A}(\bm{\bm{\Phi}}_{t_2}^{\uparrow}) \right\rangle\right) \cdot |\bm{\mathcal{M}}_{{t_1,t_2}}^{\downarrow}| \\
&+ \frac{1}{2N_{\uparrow}!}|\bm{\mathcal{M}}_{{t_1,t_2}}^{\uparrow}| \cdot \left(\sum_{i=1}^{N_{\downarrow}}\left\langle \nabla_{\tau_{\downarrow}(i)}\mathcal{A}(\bm{\Phi}_{t_1}^{\downarrow})|\nabla_{\tau_{\downarrow}(i)} \mathcal{A}(\bm{\Phi}_{t_2}^{\downarrow})\right\rangle\right),
\end{aligned}
\end{equation}
where 
\begin{equation*}
\begin{aligned}
&\sum_{i=1}^{N_{\uparrow}} \langle \nabla_{\tau_{\uparrow}(i)}\mathcal{A}(\bm{\Phi}_{t_1}^{\uparrow})|\nabla_{\tau_{\uparrow}(i)} \mathcal{A}(\bm{\Phi}_{t_2}^{\uparrow}) \rangle \\
= &\sum_{i=1}^{N_{\uparrow}} \frac{1}{N_\uparrow !} \underbrace{\begin{vmatrix}\langle\phi_{{t_1},\tau_{\uparrow}(1)} | \phi_{{t_2},\tau_{\uparrow}(1)} \rangle  & \dots & \langle\nabla\phi_{{t_1},\tau_{\uparrow}(1)} | \nabla\phi_{{t_2},\tau_{\uparrow}(i)} \rangle&\dots &\langle\phi_{{t_1},\tau_{\uparrow}(1)} | \phi_{{t_2},\tau_{\uparrow}(N_{\uparrow})} \rangle \\\vdots&  & \vdots& & \vdots\\\langle\phi_{{t_1},\tau_{\uparrow}(N_{\uparrow})} | \phi_{{t_2},\tau_{\uparrow}(1)} \rangle  & \dots &\langle\nabla\phi_{{t_1},\tau_{\uparrow}(N_{\uparrow})} |\nabla \phi_{{t_2},\tau_{\uparrow}(i)} \rangle &\dots &\langle\phi_{{t_1},\tau_{\uparrow}(N_{\uparrow})} | \phi_{{t_2},\tau_{\uparrow}(N_{\uparrow})} \rangle\end{vmatrix}}_{\bm{|\mathcal{K}}_{{t_1},{t_2},\tau_{\uparrow}(i)}^{\uparrow}|}.
\end{aligned}
\end{equation*}
Then we can use low rank update technique to compute $|\mathcal{K}_{{t_1},{t_2},\tau_{\uparrow}(i)}^{\uparrow}|$ efficiently,  
\begin{equation}
\label{eq::K_up}
\begin{aligned}
&\left|\bm{\mathcal{K}}_{t_1,t_2,\tau_{\uparrow}(i)}^{\uparrow}\right|\\
= &\int{
\begin{vmatrix}
\phi_{t_1,\tau_{\uparrow}(1)}  \phi_{t_2,\tau_{\uparrow}(1)}   & \dots & \nabla \phi_{t_1,\tau_{\uparrow}(1)} \nabla \phi_{t_2,\tau_{\uparrow}(i)} &\dots &\phi_{t_1,\tau_{\uparrow}(1)}  \phi_{t_2,\tau_{\uparrow}(N_{\uparrow})} \\
\vdots&  & \vdots& & \vdots\\
\phi_{t_1,\tau_{\uparrow}(N_{\uparrow})}  \phi_{t_2,\tau_{\uparrow}(1)}   & \dots &\nabla \phi_{t_1,\tau_{\uparrow}(N_{\uparrow})}\nabla \phi_{t_2,\tau_{\uparrow}(i)}  &\dots &\phi_{t_1,\tau_{\uparrow}(N_{\uparrow})}  \phi_{t_2,\tau_{\uparrow}(N_{\uparrow})} 
\end{vmatrix}
}\mathrm{d}\Omega^{\uparrow} \\
= &|\bm{\mathcal{M}}_{t_1,t_2}^{\uparrow}| \int{
\begin{vmatrix}
\xi_{t_1,t_2,\tau_{\uparrow}(1)}  \phi_{t_2,\tau_{\uparrow}(1)}   & \dots & \nabla \xi_{t_1,t_2,\tau_{\uparrow}(1)}\nabla \phi_{t_2,\tau_{\uparrow}(i)} &\dots &\xi_{t_1,t_2,\tau_{\uparrow}(1)}  \phi_{t_2,\tau_{\uparrow}(N_{\uparrow})} \\
\vdots&  & \vdots& & \vdots\\
\xi_{t_1,t_2,\tau_{\uparrow}(N_{\uparrow})}  \phi_{t_2,\tau_{\uparrow}(1)}   & \dots &\nabla \xi_{t_1,t_2,\tau_{\uparrow}(N_{\uparrow})}\nabla \phi_{t_2,\tau_{\uparrow}(i)}  &\dots &\xi_{t_1,t_2,\tau_{\uparrow}(N_{\uparrow})}  \phi_{t_2,\tau_{\uparrow}(N_{\uparrow})} 
\end{vmatrix}
}\mathrm{d}\Omega^{\uparrow} \\
=& |\bm{\mathcal{M}}_{t_1,t_2}^{\uparrow}| \int{
\begin{vmatrix}
1   & \dots & \nabla \xi_{t_1,t_2,\tau_{\uparrow}(i)}\nabla \phi_{t_2,\tau_{\uparrow}(i)}) &\dots &0 \\
\vdots&  & \vdots& & \vdots\\
0    & \dots &\nabla \xi_{t_1,t_2,\tau_{\uparrow}(N_{\uparrow})}\nabla \phi_{t_2,\tau_{\uparrow}(i)}  &\dots &1 
\end{vmatrix}
}\mathrm{d}\bm{r}_{\tau_{\uparrow}(i)} \\
=& |\bm{\mathcal{M}}_{t_1,t_2}^{\uparrow}| \langle  \nabla \xi_{t_1,t_2,\tau_{\uparrow}(i)} \mid  \nabla \phi_{t_2,\tau_{\uparrow}(i)} \rangle,
\end{aligned}
\end{equation}
where $\mathrm{d}\Omega^{\uparrow}$ denotes $\mathrm{d}\bm{r}_{\tau_{\uparrow}(1)}\cdots \mathrm{d}\bm{r}_{\tau_{\uparrow}(N_{\uparrow})}$. The computation of the spin-down part is analogous to Eq.~\eqref{eq::K_up}, which provides that
\begin{equation}
\label{eq::S_K}
\frac{1}{2}\sum_i^N \langle \nabla_i\Psi | \nabla_i \Psi \rangle_{(t_1,t_2)} = \frac{1}{2}\left\langle\Psi|\Psi\right\rangle_{t_1,t_2} \sum_{i}^N  \langle \nabla_i \xi_{t_1,t_2,i}|\nabla_i \phi_{t_2,i}\rangle.
\end{equation}
Following the similar procedure of $\left(E_{\text{K}}\right)_{t_1,t_2}$, the ion-electron interaction $\left(E_{ie}\right)_{t_1,t_2}$ can also be regarded as an one-body operator  with
\begin{equation} 
\begin{aligned}
&- \sum_k^M\sum_i^N \left \langle \Psi\left|\frac{Q_k}{|\bm{r}_i-\bm{R}_k|}\right|\Psi \right\rangle_{t_1,t_2} \\= & -\sum_{k}^M \frac{1}{N_{\downarrow}!}\left(\sum_{i=1}^{N_{\uparrow}} \langle \mathcal{A}(\bm{\Phi}_{t_1}^{\uparrow})\left|\frac{Q_k}{|\bm{r}_{\tau_{\uparrow}(i)}-\bm{R}_k|}\right|\mathcal{A}(\bm{\Phi}_{t_2}^{\uparrow}) \rangle\right) \cdot |\bm{\mathcal{M}}_{{t_1},{t_2}}^{\downarrow}|  \\&- \sum_{k}^M \frac{1}{N_{\uparrow}!}|\bm{\mathcal{M}}_{{t_1},{t_2}}^{\uparrow}| \cdot \left(\sum_{i=1}^{N_{\downarrow}}\langle \mathcal{A}(\bm{\Phi}_{t_1}^{\downarrow})\left|\frac{Q_k}{|\bm{r}_{\tau_{\downarrow}(i)}-\bm{R}_k|}\right| \mathcal{A}(\bm{\Phi}_{t_2}^{\downarrow})\rangle\right),
\end{aligned}
\end{equation}
where 
\begin{equation}
\label{eq::A_ie}
\begin{aligned}
& \sum_{i=1}^{N_{\uparrow}} \left \langle \mathcal{A}(\bm{\Phi}_s^{\uparrow})\left|\frac{Q_k}{|\bm{r}_{\tau_{\uparrow}(i)}-\bm{R}_k|}\right|\mathcal{A}(\bm{\Phi}_t^{\uparrow}) \right \rangle\\ =& \sum_{i=1}^{N_{\uparrow}} \left \langle \mathcal{A}(\bm{\Phi}_s^{\uparrow})\left|\frac{Q_k}{|\bm{r}_{\tau_{\uparrow}(i)}-\bm{R}_k|}\right|\prod_i^{N_\uparrow} \phi_{{t_2},\tau_{\uparrow}(i)} \right \rangle \\
=& |\bm{\mathcal{M}}_{s,t}^{\uparrow}| \sum_{i=1}^{N_{\uparrow}} \left \langle \mathcal{A}(\bm{\Xi}_{s,t}^{\uparrow})\left|\frac{Q_k}{|\bm{r}_{\tau_{\uparrow}(i)}-\bm{R}_k|}\right|\prod_i^{N_\uparrow} \phi_{{t_2},\tau_{\uparrow}(i)} \right \rangle \\
=& \frac{|\bm{\mathcal{M}}_{s,t}^{\uparrow}|}{N_\uparrow !}\sum_{i=1}^{N_{\uparrow}}\int{\begin{vmatrix}
\xi_{s,t,{\tau_{\uparrow}(1)}}    & \dots & \xi_{s,t,{\tau_{\uparrow}(1)}}\left|\frac{Q_k}{|\bm{r}_{\tau_{\uparrow}(i)}-\bm{R}_k|}\right|  &\dots &\xi_{s,t,{\tau_{\uparrow}(1)}}  \\
\vdots&  & \vdots& & \vdots\\
\xi_{s,t,{\tau_{\uparrow}(N_{\uparrow})}}   & \dots &\xi_{s,t,{\tau_{\uparrow}(N_{\uparrow})}}\left|\frac{Q_k}{|\bm{r}_{\tau_{\uparrow}(i)}-\bm{R}_k|}\right|    &\dots &\xi_{s,t,{\tau_{\uparrow}(N_{\uparrow})}}  
\end{vmatrix}} \prod_{i=1}^{N_{\uparrow}} \phi_{t,{\tau_{\uparrow}(i)}} \mathrm{d}\Omega^{\uparrow} \\
=&\frac{|\bm{\mathcal{M}}_{s,t}^{\uparrow}|}{N_\uparrow !} \sum_{i=1}^{N_{\uparrow}}\int{\begin{vmatrix}
1   & \dots & \xi_{s,t,{\tau_{\uparrow}(1)}}\left|\frac{Q_k}{|\bm{r}_{\tau_{\uparrow}(i)}-\bm{R}_k|}\right|  \phi_{t,{\tau_{\uparrow}(i)}} &\dots &0 \\
\vdots&  & \vdots& & \vdots\\
0    & \dots &\xi_{s,t,{\tau_{\uparrow}(N_{\uparrow})}}\left|\frac{Q_k}{|\bm{r}_{\tau_{\uparrow}(i)}-\bm{R}_k|}\right| \phi_{t,{\tau_{\uparrow}(i)}}  &\dots &1 
\end{vmatrix}}\mathrm{d}\bm{r}_{\tau_{\uparrow}(i)} \\
=& \frac{|\bm{\mathcal{M}}_{s,t}^{\uparrow}|}{N_\uparrow !} \sum_{i=1}^{N_{\uparrow}}\left \langle  \xi_{s,t,{\tau_{\uparrow}(i)}}\left|\frac{Q_k}{|\bm{r}_{\tau_{\uparrow}(i)}-\bm{R}_k|}\right| \phi_{t,{\tau_{\uparrow}(i)}}\right \rangle.
\end{aligned}
\end{equation}
The first equality in Eq.~\eqref{eq::A_ie} holds because the antisymmetrization operator $\mathcal{A}$ is a linear, orthogonal, and idempotent projection (see \cite{beylkin2008approximating}). Indeed, the spin-down part can be simplified in a similar manner, which provides
\begin{equation}
\label{eq::S_ie}
- \sum_k^M\sum_i^N \left \langle \Psi\left|\frac{Q_k}{|\bm{r}_i-\bm{R}_k|}\right|\Psi \right\rangle_{t_1,t_2} =- \left\langle\Psi|\Psi\right\rangle_{t_1,t_2} \sum_k^M\sum_{i}^N  \left\langle \xi_{t_1,t_2,i}\left|\frac{Q_k}{|\bm{r}_i-\bm{R}_k|}\right| \phi_{t_2,i} \right\rangle.
\end{equation}
Finally, we cosider the element-wise computation of electron-electron interaction $\left(E_{ee}\right)_{t_1,t_2}$,
\begin{equation}
\begin{aligned}
\sum_{i<j}\left\langle\Psi\left|\frac{1}{|\bm{r}_{i}-\bm{r}_j|}\right|\Psi\right\rangle_{t_1,t_2} &= \sum_{1\le i<j \le N_{\uparrow}}\left\langle\Psi\left|\frac{1}{|\bm{r}_{\tau_{\uparrow}(i)}-\bm{r}_{\tau_{\uparrow}(j)}|}\right|\Psi\right\rangle_{t_1,t_2} 
\\ &+ \sum_{1\le i \le N_{\uparrow}, 1 \le j \le N_{\downarrow}}\left\langle\Psi\left|\frac{1}{|\bm{r}_{\tau_{\uparrow}(i)}-\bm{r}_{\tau_{\downarrow}(j)}|}\right|\Psi\right\rangle_{t_1,t_2} 
\\ &+ \sum_{1 \le i<j \le N_{\downarrow}}\left\langle\Psi\left|\frac{1}{|\bm{r}_{\tau_{\downarrow}(i)}-\bm{r}_{\tau_{\downarrow}(j)}|}\right|\Psi\right\rangle_{t_1,t_2}.
\end{aligned}
\end{equation}
For the same-spin case, we can still exploit the low-rank update of determinant to fast evaluation,
\begin{equation*}
\begin{aligned}
&\left \langle \Psi\left|\frac{1}{|\bm{r}_{\tau_{\uparrow}(i)}-\bm{r}_{\tau_{\uparrow}(j)}|}\right| \Psi \right \rangle_{({t_1},{t_2})}\\ =& \frac{|\bm{\mathcal{M}}_{{t_1},{t_2}}^{\downarrow}|}{N_{\downarrow}!} \left \langle \mathcal{A}(\bm{\Phi}_{t_1}^{\uparrow})\left|\frac{1}{|\bm{r}_{\tau_{\uparrow}(i)}-\bm{r}_{\tau_{\uparrow}(j)}|}\right|\mathcal{A}(\bm{\Phi}_{t_2}^{\uparrow}) \right\rangle \\
=& \frac{|\bm{\mathcal{M}}_{{t_1},{t_2}}^{\downarrow}|}{N_{\downarrow}!} \left \langle \mathcal{A}(\bm{\Phi}_{t_1}^{\uparrow})\left|\frac{1}{|\bm{r}_{\tau_{\uparrow}(i)}-\bm{r}_{\tau_{\uparrow}(j)}|}\right|\prod_i^{N_\uparrow} \phi_{{t_2},\tau_{\uparrow}(i)}\right\rangle\\
=& \frac{|\bm{\mathcal{M}}_{{t_1},{t_2}}^{\uparrow}| \cdot |\bm{\mathcal{M}}_{{t_1},{t_2}}^{\downarrow}|}{ N_{\downarrow}!} \left \langle \mathcal{A}(\bm{\Xi}_{{t_1},{t_2}}^{\uparrow})\left|\frac{1}{|\bm{r}_{\tau_{\uparrow}(i)}-\bm{r}_{\tau_{\uparrow}(j)}|}\right|\prod_i^{N_\uparrow} \phi_{{t_2},\tau_{\uparrow}(i)}\right\rangle\\
=&  \int{\frac{\left\langle\Psi|\Psi\right\rangle_{t_1,t_2}}{|\bm{r}_{\tau_{\uparrow}(i)}-\bm{r}_{\tau_{\uparrow}(j)}|}\begin{vmatrix}
\xi_{{t_1},{t_2},\tau_{\uparrow}(1)}    & \dots & \xi_{{t_1},{t_2},\tau_{\uparrow}(1)}  &\dots &\xi_{{t_1},{t_2},\tau_{\uparrow}(1)}  \\
\vdots&  & \vdots& & \vdots\\
\xi_{{t_1},{t_2},\tau_{\uparrow}(N_{\uparrow})}   & \dots &\xi_{{t_1},{t_2},\tau_{\uparrow}(N_{\uparrow})}    &\dots &\xi_{{t_1},{t_2},\tau_{\uparrow}(N_{\uparrow})}  
\end{vmatrix}} \prod_i^{N_\uparrow} \phi_{{t_2},\tau_{\uparrow}(i)} \mathrm{d}\Omega^{\uparrow} \\
=& \int{\frac{\left\langle\Psi|\Psi\right\rangle_{t_1,t_2} }{|\bm{r}_{\tau_{\uparrow}(i)}-\bm{r}_{\tau_{\uparrow}(j)}|}\begin{vmatrix}
1  & \dots & \xi_{{t_1},{t_2},\tau_{\uparrow}(1)}\phi_{{t_2},\tau_{\uparrow}(i)} & \dots & \xi_{{t_1},{t_2},\tau_{\uparrow}(1)}\phi_{{t_2},\tau_{\uparrow}(j)} &\dots &0 \\
\vdots  &  & \vdots &  & \vdots & &\vdots\\
0  & \dots & \xi_{{t_1},{t_2},{\tau_{\uparrow}(N_{\uparrow})}}\phi_{{t_2},\tau_{\uparrow}(i)} & \dots & \xi_{{t_1},{t_2},{\tau_{\uparrow}(N_{\uparrow})}}\phi_{{t_2},\tau_{\uparrow}(j)}&\dots & 1
\end{vmatrix}}\mathrm{d}\bm{r}_{\tau_{\uparrow}(i)}\mathrm{d}\bm{r}_{\tau_{\uparrow}(j)} \\
=&  \left\langle\Psi|\Psi\right\rangle_{t_1,t_2}\left \langle 
\begin{vmatrix}
\xi_{{t_1},{t_2},\tau_{\uparrow}(i)}(r_i)  & \xi_{{t_1},{t_2},\tau_{\uparrow}(i)}(r_j)\\
\xi_{{t_1},{t_2},\tau_{\uparrow}(j)}(r_i)  & \xi_{{t_1},{t_2},\tau_{\uparrow}(j)}(r_j)
\end{vmatrix}\left|\frac{1}{|\bm{r}_{\tau_{\uparrow}(i)}-\bm{r}_{\tau_{\uparrow}(j)}|}\right|\phi_{{t_2},\tau_{\uparrow}(i)}\phi_{{t_2},\tau_{\uparrow}(j)}
\right\rangle.
\end{aligned} 
\end{equation*}
The spin-down part of the same-spin case is analogous. For the opposite-spin case, it admits a simpler form that
\begin{equation*}
\begin{aligned}
&\left \langle \Psi \left|\frac{1}{|\bm{r}_{\tau_{\uparrow}(i)}-\bm{r}_{\tau_{\downarrow}(j)}|}\right| \Psi \right \rangle_{(t_1,{t_2})}\\ =& \left \langle \mathcal{A}(\bm{\Phi}_{t_1}^{\uparrow}) \cdot \mathcal{A}(\bm{\Phi}_{t_1}^{\downarrow}) \left|\frac{1}{|\bm{r}_{\tau_{\uparrow}(i)}-\bm{r}_{\tau_{\downarrow}(j)}|}\right| \mathcal{A}(\bm{\bm{\Phi}}_{t_2}^{\uparrow}) \cdot \mathcal{A}(\bm{\Phi}_{t_2}^{\downarrow}) \right \rangle \\
=& \left \langle \mathcal{A}(\bm{\Phi}_{t_1}^{\uparrow}) \cdot \mathcal{A}(\bm{\Phi}_{t_1}^{\downarrow}) \left|\frac{1}{|\bm{r}_{\tau_{\uparrow}(i)}-\bm{r}_{\tau_{\downarrow}(j)}|}\right| \prod_i^{N_\uparrow} \phi_{{t_2},\tau_{\uparrow}(i)} \cdot \prod_j^{N_\downarrow} \phi_{{t_2},\tau_{\downarrow}(j)} \right \rangle \\
=& |\bm{\mathcal{M}}_{{t_1},{t_2}}^{\uparrow}| |\bm{\mathcal{M}}_{{t_1},{t_2}}^{\downarrow}|\left \langle \mathcal{A}(\bm{\Xi}_{{t_1},{t_2}}^{\uparrow}) \cdot \mathcal{A}(\bm{\Xi}_{{t_1},{t_2}}^{\downarrow}) \left|\frac{1}{|\bm{r}_{\tau_{\uparrow}(i)}-\bm{r}_{\tau_{\downarrow}(j)}|}\right|\prod_i^{N_\uparrow} \phi_{{t_2},\tau_{\uparrow}(i)} \cdot \prod_j^{N_\downarrow} \phi_{{t_2},\tau_{\downarrow}(j)} \right \rangle \\
=& \int{\frac{\left\langle\Psi|\Psi\right\rangle_{t_1,t_2} }{|\bm{r}_{\tau_{\uparrow}(i)}-\bm{r}_{\tau_{\downarrow}(j)}|} }\begin{vmatrix}
1   & \dots & \xi_{{t_1},{t_2},{\tau_{\uparrow}(1)}}  \phi_{{t_2},{\tau_{\uparrow}(i)}} &\dots &0 \\
\vdots&  & \vdots& & \vdots\\
0    & \dots &\xi_{{t_1},{t_2},{\tau_{\uparrow}(N_{\uparrow})}} \phi_{{t_2},{\tau_{\uparrow}(i)}}  &\dots &1 
\end{vmatrix}\begin{vmatrix}
1   & \dots & \xi_{{t_1},{t_2},{\tau_{\downarrow}(1)}}  \phi_{{t_2},{\tau_{\downarrow}(j)}} &\dots &0 \\
\vdots&  & \vdots& & \vdots\\
0    & \dots &\xi_{{t_1},{t_2},{\tau_{\downarrow}(N_{\downarrow})}} \phi_{{t_2},{\tau_{\downarrow}(j)}}  &\dots &1 
\end{vmatrix} \mathrm{d}\bm{r}_{\tau_{\uparrow}(i)} \mathrm{d}\bm{r}_{\tau_{\downarrow}(j)} \\
=& \left\langle\Psi|\Psi\right\rangle_{t_1,t_2}\left \langle \xi_{{t_1},{t_2},\tau_{\uparrow}(i)}\xi_{{t_1},{t_2},\tau_{\downarrow}(j)}\left|\frac{1}{|\bm{r}_{\tau_{\uparrow}(i)}-\bm{r}_{\tau_{\downarrow}(j)}|}\right|\phi_{{t_2},\tau_{\uparrow}(i)}\phi_{{t_2},\tau_{\downarrow}(j)} \right \rangle 
\end{aligned}
\end{equation*}
Combining the above cases yields
\begin{equation}
\label{eq::S_ee}
\begin{aligned}
&\left\langle\Psi\left|\frac{1}{|\bm{r}_i-\bm{r}_j|}\right|\Psi\right\rangle_{t_1,t_2}\\
=&\left\{\begin{aligned}
&\left\langle\Psi|\Psi\right\rangle_{t_1,t_2}\left\langle\begin{vmatrix}
\xi_{t_1,t_2,i}(\bm{r}_i)  & \xi_{t_1,t_2,i}(\bm{r}_j)\\
\xi_{t_1,t_2,j}(\bm{r}_i)  & \xi_{t_1,t_2,j}(\bm{r}_j)
\end{vmatrix}\left|\frac{1}{|\bm{r}_i-\bm{r}_j|}\right|\phi_{t_2,i}\phi_{t_2,j}\right\rangle,\quad \text{$i,j$ same spin}\\
&\left\langle\Psi|\Psi\right\rangle_{t_1,t_2}\left\langle\xi_{t_1,t_2,i}\xi_{t_1,t_2,j}\left|\frac{1}{|\bm{r}_i-\bm{r}_j|}\right|\phi_{t_2,i}\phi_{t_2,j}\right\rangle,\quad \text{otherwise}.\\
\end{aligned}
\right.
\end{aligned}
\end{equation}
Therefore, with Eq.~\eqref{eq::S_K}, Eq.~\eqref{eq::S_ie} and Eq.~\eqref{eq::S_ee}, we finish the derivation of Eqs.~\eqref{eq::Slater_all} and \eqref{eq::Slater_ee} in the main manuscript.

% \bibliographystyle{elsart-num-sort}
% \bibliography{SE}

\end{document}